\documentclass[aps,prl,twocolumn,superscriptaddress,nofootinbib]{revtex4-1}
\pdfoutput=1

\usepackage{amsmath}
\usepackage{amsthm}
\usepackage{amssymb}
\usepackage{mathrsfs}
\usepackage{bm, dsfont}
\usepackage[usenames,dvipsnames]{color}
\usepackage{colortbl}
\usepackage[table]{xcolor}
\usepackage{enumitem}

\usepackage{graphicx}
\usepackage{natbib}
\usepackage[colorlinks=true,linkcolor=blue,citecolor=blue,urlcolor=blue]{hyperref}

\usepackage{hypcap}
\usepackage{verbatim, float}
\usepackage{psfrag}
\usepackage[normalem]{ulem}
\usepackage{physics}

\newcommand{\be}{\begin{equation}}
\newcommand{\ee}{\end{equation}}

\newtheorem*{observation}{Observation}

\newtheorem{theorem}{Theorem}

\newtheorem{definition}[theorem]{Definition}
\newtheorem{corollary}[theorem]{Corollary}
\newtheorem{lemma}[theorem]{Lemma}
\newtheorem{proposition}[theorem]{Proposition}
\newtheorem{example}[theorem]{Example}

\theoremstyle{remark}

\newcommand{\hA}{\mathcal{A}}
\newcommand{\hB}{\mathcal{B}}

\newcommand{\hD}{\mathcal{D}}
\newcommand{\hE}{\mathcal{E}}
\newcommand{\hF}{\mathcal{F}}

\newcommand{\hI}{\mathcal{I}} 
\newcommand{\hJ}{\mathcal{J}}
\newcommand{\hK}{\mathcal{K}}
\newcommand{\hL}{\mathcal{L}}
\newcommand{\hM}{\mathcal{M}}
\newcommand{\hN}{\mathcal{N}}
\newcommand{\hO}{\mathcal{O}} 

\newcommand{\hS}{\mathcal{S}}
\newcommand{\hT}{\mathcal{T}}

\newcommand{\N}{\mathbb N} 
\newcommand{\R}{\mathbb R} 

\newcommand{\Q}{\mathbb Q} 

\newcommand{\kb}[2]{|#1 \rangle\langle #2|}

\newcommand{\fii}{\varphi} 
\newcommand{\om}{\omega} 

\def\<{\langle} 
\def\>{\rangle}
\newcommand{\hi}{\mathcal{H}} 
\newcommand{\ki}{\mathcal{K}} 

\newcommand{\F}{\mathcal F} 

\newcommand{\lh}{\mathcal{L(H)}} 
\newcommand{\kh}{\mathcal{L(K)}} 
\renewcommand{\th}{\mathcal{T(H)}} 
\newcommand{\sh}{\mathcal{S(H)}} 
\newcommand{\lhs}{{\mathcal L}_s(\hi)}

\newcommand{\fix}[1]{\mathrm{Fix}\left( #1 \right)}

\begin{document}

\title{Generalised contextuality of continuous variable quantum theory can be revealed with a single projective measurement}

\author{Pauli Jokinen}
\affiliation{Department of Physics and Astronomy, Uppsala University, Box 516, 751 20 Uppsala, Sweden}
\affiliation{Nordita, KTH Royal Institute of Technology and Stockholm University, 10691 Stockholm, Sweden}
\email{pauli.jokinen@physics.uu.se}
\author{Mirjam Weilenmann}
\affiliation{Inria, T\'el\'ecom Paris - LTCI, Institut Polytechnique de Paris, 91120 Palaiseau, France}
\affiliation{Department of Applied Physics, University of Geneva, Switzerland}
\author{Martin Plávala}
\affiliation{Institut f\"ur Theoretische Physik, Leibniz Universit\"at Hannover, 30167 Hannover, Germany}
\author{Juha-Pekka Pellonpää}
\affiliation{Department of Physics and Astronomy, University of Turku, FI-20014 Turun yliopisto, Finland}
\author{Jukka Kiukas}
\affiliation{Department of Mathematics, Aberystwyth University, Aberystwyth SY23 3BZ, United Kingdom}
\author{Roope Uola}
\affiliation{Department of Physics and Astronomy, Uppsala University, Box 516, 751 20 Uppsala, Sweden}
\affiliation{Nordita, KTH Royal Institute of Technology and Stockholm University, 10691 Stockholm, Sweden}

\begin{abstract}
  Generalized contextuality is a possible indicator of non-classical behaviour in quantum information theory. In finite-dimensional systems, this is justified by the fact that noncontextual theories can be embedded into some simplex, i.e. into a classical theory. We show that a direct application of the standard definition of generalized contextuality to continuous variable systems does not envelope the statistics of some basic measurements, such as the position observable. In other words, we construct families of fully classical, i.e. commuting, measurements that nevertheless can be used to show contextuality of quantum theory. To overcome the apparent disagreement between the two notions of classicality, that is commutativity and noncontextuality, we propose a modified definition of generalised contextuality for continuous-variable systems. The modified definition is based on a physically-motivated approximation procedure, that uses only finite sets of measurement effects. We prove that in the limiting case this definition corresponds exactly to an extension of noncontextual models that benefits from non-constructive response functions. In the process, we discuss the extension of a known connection between contextuality and no-broadcasting to the continuous-variable scenario, and prove structural results regarding fixed points of infinite-dimensional entanglement breaking channels.
\end{abstract}

\maketitle

\section{I. Introduction}

The ongoing second quantum revolution aims to provide various advances from cryptographic possibilities to advanced materials, and from quantum computational speed-ups to unprecedented metrological precisions. This rises a natural question on characterizing the quantum resources that fuel such development. For example, entanglement, incompatibility, and coherence are known to be necessary for a quantum advantage in various tasks \cite{piani09,quintino14,uola14,Napoli_16,Takagi_2019,carmeli19a,skrzypczyk19,uola19b,oszmaniec19,xu19, lami2021framework,haapasalo2021operational,aubrun2022entanglement,Guhne2023JMreview}, but one can find tasks where these properties are not sufficient \cite{werner89, hirsch18,selby21}. It is, hence, not to be expected that there would exist a universal non-classical entity, that could characterize the quantum advantage in full.

During the last two decades, however, the notion of generalised contextuality \cite{spekkens05}  has risen as a candidate for enveloping the non-classicality of a vast class of theoretical entities \cite{Spekkens_2008,Pusey_2014,Lostaglio_2018} and quantum information processing tasks \cite{SpekkensRAC,Schmid_2018,Lostaglio_2020}. The notion of generalised contextuality is arguably very flexible, in that it is applicable to both single and compound systems, unlike for example entanglement, and the so-called ontological models underlying the concept can often be adjusted to find not only necessary, but also sufficient conditions for various classical-to-quantum boundaries, such as incompatibility, entanglement, steering, and broadcasting \cite{tavakoli19,plavala2022incompatibility,wright2023invertible,plavala2024contextuality,jokinen2024nobroadcasting}. On top of the broad applicability and foundationally intriguing structure, the concept is even sometimes called the “gold standard" of classical explainability \cite{Schmid24addressing,Schmid_2025}. 
While research exploring the connections between contextuality and other quantum notions is thus abundant, extensions of these connections to continuous variable systems have so far been sparse \cite{Plastino2010,McKeown2011,Su2012,Asadian2015,LaversanneFinot2017,Barbosa2022,Booth2022,delgado2025frame}. Nevertheless, in order to put the role of generalised contextuality to the test, these are arguably inevitable.

In this work, we show that a direct application of generalised contextuality to continuous variable systems leads one to question the role of contextuality as a notion of classicality. Namely, we demonstrate that one can prove the contextuality of quantum theory from a single position measurement. As this measurement is commutative, we deem it to be classical: its measurement data can be presented in a classical function space. We propose a way around this discrepancy by reconsidering the probability-theoretical foundations of generalized contextuality in continuous variable systems. We introduce an approximate notion of noncontextuality and show that it is equivalent to the existence of a non-normal ontological model. Such model is supported by non-constructive response functions. With such models, we show that the data of a position measurement does not reveal the contextual nature of quantum theory.
In the continuous variable case, one thus has to either give up the natural formulation of noncontextuality in terms of probabilistic response functions, or, accept that the position measurement can confirm contextuality.

We further show that the position example is not the only one leading to such behaviour. By extending recent results linking contextuality and the no-broadcasting theorem \cite{jokinen2024nobroadcasting} to the continuous variable regime, we show an interplay among the original definition, the non-normal definition, and representability using commuting algebras. As a side result, we characterize the fixed points of Schrödinger picture entanglement breaking channels as commuting sets of states, and give a necessary condition for the fixed points of the Heisenberg picture through broadcastability.

The paper is structured as follows. In Section II, we recall the definition of generalized contextuality in finite-dimensional systems. In Section III, we show how such definition leads to the possibility of confirming contextuality of quantum theory from a single commuting measurement. We further propose two relaxations of the noncontextual models, approximative and non-normal,
that envelope the position measurement, and show that these relaxations are equivalent. In Section IV, we investigate the connection between continuous variable contextuality and broadcasting. We also give an algebraic formulation for the normal case. In Section V, we provide the mathematical proofs for the results concerning the approximate and non-normal models.

\section{II. Generalized contextuality in finite-dimensional discrete systems} 
We begin by briefly reviewing the basic framework of generalized contextuality and its application to finite-dimensional quantum systems.  \par 
The notion of generalized contextuality was first introduced in \cite{spekkens05}. A core building block of generalized contextuality is the concept of an abstract operational theory, consisting of preparations $\{P\}$, measurements $\{M\}$ and the probabilities $p(k|M,P)$ predicted by the theory. Here $k$ denotes a (discrete) measurement outcome. Occasionally also transformation procedures of the operational theory are examined, but in this work we focus solely on preparations and measurements. \par 
Within an operational theory one can define the so-called \emph{operational equivalence classes}. Two preparations $P,P'$ are operationally equivalent if the following condition holds.
\begin{align}\label{eqprepnoncontext}
    p(k|M,P)=p(k|M,P') \quad \forall M, \; \forall k
\end{align}
Similarly the measurements $M, M'$ are equivalent if the following condition holds.
\begin{align}\label{eqmeasnoncontext}
    p(k|M,P)=p(k|M',P) \quad \forall P, \; \forall k
\end{align}
The probabilities of the operational theory are then said to follow a \emph{preparation- and measurement noncontextual ontological model} (PNC+MNC) if 
\begin{align}\label{eqontmodel}
    p(k|M,P)=\sum_\lambda p(\lambda|P)p(k|M,\lambda),
\end{align}
where the response functions (probabilities) $p(\lambda|P)$ and $p(k|M,\lambda)$ depend only on the equivalence classes defined by equations \eqref{eqprepnoncontext} and \eqref{eqmeasnoncontext} respectively for all $\lambda$. Furthermore $p(\lambda|P)$ and $p(k|M,\lambda)$ are convex linear in some abstract convex structures of the preparations and measurements. An operational theory for which there exist no explanation of the form \eqref{eqontmodel} is called \emph{preparation- and measurement contextual}.\par 
A concrete example of an operational theory is quantum theory. In this case, the preparation equivalence classes defined by \eqref{eqprepnoncontext} are identified with quantum states $\rho$, i.e.\ trace-1 positive-semidefinite linear operators. Similarly, the measurement equivalence classes defined by \eqref{eqmeasnoncontext} are identified with positive-operator-valued measures (POVMs) $M$, i.e. collections of positive-semidefinite operators $\{M_k\}_k$ normalizing to the identity operator: $\sum_k M_k=I$. Furthermore, the probabilities are produced according to the Born rule: $p(k|M,\rho)=\tr{\rho M_k}$.\par  For finite-dimensional quantum theory, the ontological model \eqref{eqontmodel} takes a simpler form due to the convex linearity of the response functions, and Hilbert-Schmidt duality of a finite-dimensional Hilbert space $\hi$:
\begin{align}\label{eqontmodelquant}
    \tr{\rho M_k }=\sum_\lambda \tr{G_\lambda \rho} \tr{M_k \sigma_\lambda},
\end{align}
where the $\{G_\lambda \}_\lambda$ satisfy $\tr{G_\lambda \rho} \geq 0 \ \forall \lambda$,  $\sum_{\lambda}\tr{G_\lambda \rho}=1$ and $\tr{M_k \sigma_\lambda}\geq 0 \ \forall \lambda$.
In this work we will consider the full quantum theory as our operational theory i.e. the set of all quantum states, denoted $\sh$, and collection of all observables, denoted $\hO$. In this case $\{G_\lambda\}_\lambda$ in equation \eqref{eqontmodelquant} is a POVM and $\sigma_\lambda \in \sh$ for all $\lambda$. With this, an entanglement breaking channel (EBC),  $\Lambda_{EB}(T)=\sum_{\lambda}\tr{G_\lambda T}\sigma_\lambda$,  can be identified from equation \eqref{eqontmodelquant}: $\tr{\rho M_k}=\tr{\Lambda_{EB}(\rho)M_k}$. A quantum channel is a completely positive trace-preserving linear map and an EBC is a special case of this. Therefore, the question of contextuality reduces to a problem of fixed points of entanglement breaking channels. This condition was characterized in our earlier work \cite{jokinen2024nobroadcasting}.

\section{III. Generalized contextuality in continuous-variable systems}
We now intend to establish the formulation of generalized contextuality presented in the previous section to the continuous-variable regime in the context of quantum theory. From now on $\hi$ and $\ki$ denote separable, possibly infinite-dimensional Hilbert spaces unless otherwise stated.
\par 
As in the finite-dimensional case, the preparation equivalence classes are identified with quantum states and measurement equivalence classes are identified with POVMs. Since we now consider continuous-variable systems, we need to slightly extend the definition of a discrete POVM: instead of single outcomes $k$ we in general consider subsets $X$ of the set of all possible outcomes $\Omega$. The possible outcome sets are called \emph{measurable} sets, which in technical terms form a $\sigma$\emph{-algebra}. A $\sigma$-algebra is a nonempty collection of subsets of $\Omega$, closed under complement and countable unions. In this work all measurable sets are assumed to belong to a \emph{standard Borel space} ($\Omega, \hB(\Omega))$. Here $\Omega$ is a complete separable metrizable space, and $\hB(\Omega)$ the $\sigma$-algebra generated by its open sets, also called the Borel $\sigma$-algebra. For example $\R^n$ and $\N$ form standard Borel spaces.  \par 
Therefore a POVM associates a positive bounded linear operator $M(X)$ to all sets of possible outcomes i.e. Borel sets $X$, analogously to discrete POVMs. Furthermore POVMs should produce probabilities according to the Born rule i.e. the map $\hB(\Omega) \ni X \mapsto \tr{\rho M(X)}$ should be a \emph{probability measure} for all $\rho \in \sh$. This means that $\tr{\rho M(X)}\ge0$, $\tr{\rho M(\Omega)}=1$ and that $X \mapsto \tr{\rho M(X)}$ is $\sigma$\emph{-additive}. $\sigma$-additivity means that the measure should decompose with respect to countable unions of disjoint outcome sets: $\tr{\rho M\left( \bigcup_{n \in \N} X_n\right)}=\sum_{n \in \N} \tr{\rho M(X_n)}$, where $X_n \cap X_m=\emptyset$ for $m\neq n$. This condition will turn out to be an interesting aspect in the definition of a noncontextual ontological model. To summarize: a POVM is a mapping $M: \hB(\Omega) \to \lh$ such that for all states $\rho \in \sh$, $X \mapsto \tr{\rho M(X)}$ is a probability measure. We shall denote from now on the measurable space in which a POVM $M$ is defined in by $(\Omega_M,\hB(\Omega_M))$. This definition of a POVM coincides with the usual definition of a finite outcome POVM, since the discrete topology (all subsets) of a finite outcome set can be metrized into a standard Borel space with the discrete metric. 
\subsection{Three definitions of continuous variable contextuality for quantum theory}
We now have the necessary tools to discuss definition of a noncontextual ontological model in the case where the operational theory is set to be quantum theory. In  analogue to equation \eqref{eqontmodel}, one could propose the following model for all $\rho \in \sh$,  $M \in \hO$ and $X \in \hB(\Omega_M)$.
\begin{align}\label{eqincontmodel}
    \tr{\rho M(X)}=\int_\Omega  p(M(X)|\lambda) d\mu_\rho(\lambda).
\end{align}
Here $\mu_\rho$ is a probability measure for all $\rho \in \sh$. Note that unlike in \cite{spekkens05}, we have here absorbed the probability density in equation \eqref{eqontmodel} into a measure, as considering the preparation responses as mappings from states to probability measures is more natural for continuous-variable systems. This is due to the fact that not all measures have a well-defined density, consider e.g. the Dirac measure in $\R$.
\par 

Now as in the finite case, $\rho \mapsto \mu_\rho$ and $E \mapsto p(E|\lambda)$ are convex linear.
Using this convex linearity and a standard Frechet-Riesz-type argument given in e.g. \cite{busch03,busch16}, we see that for all $X \in \hB(\Omega)$, $\mu_\rho(X)=\tr{\rho G(X)}$, where $G:\hB(\Omega) \to \lh$ is a POVM. 

For the measurement responses $E \mapsto p(E|\lambda)$, defined for full quantum theory, the first natural choice would be that they are given by quantum states i.e. $p(E|\lambda)=\tr{\sigma_\lambda E}$. Thus they constitute a measure-and-prepare scenario, analogous to the finite-dimensional case. We state this finite-dimensional analogue as the first definition of \emph{contextuality non-confirming sets}.

\begin{definition}\label{defnormal}
    (\textbf{Normal contextuality non-confirming sets}) \\
    Let $\hM\subset \hO$ be a set of measurements and $\hS \subset \sh$ a set of states. $\hM$ fails to prove the contextuality of quantum theory, i.e. is contextuality non-confirming, if there is a POVM $G:\hB(\Omega) \to \lh$ and a (weakly) measurable family of states $\{\sigma_\lambda\}_{\lambda \in \Omega} \subset \sh$ such that for all $X \in \hB(\Omega_M)$ we have
    \begin{align}\label{eqdefinitionnorm}
        \tr{\rho M(X)}=\int_\Omega \tr{M(X)\sigma_\lambda} \, d(\tr{\rho G(\lambda)})
    \end{align}
    for all $\rho \in \sh$ and $M \in \hM$. 
    Similarly $\hS$ is contextuality non-confirming if and only if equation \eqref{eqdefinitionnorm} holds for all $M \in \hO$ and $\rho \in \hS$.
\end{definition}
We examine the choice of measurement responses more carefully in the next section. 

 It is fairly apparent that Definition \ref{defnormal} induces a fixed point problem of EB-channels in analogy to the finite case by using Theorem 2 of \cite{holevo2005separability}. In other words the contextuality non-confirming sets of states $\hS$ are exactly the ones such that for all $\rho \in \hS$ we have $\Lambda_{EB}(\rho)=\rho$ for some EB-channel $\Lambda_{EB}$, and similarily for measurements. We relegate the details of this to the Appendix. 

The form of the measurement response functions in Definition \ref{defnormal} is also essential for their interpretation as valid probabilities as we will soon show. However, one quickly runs into arguably counterintuitive results with this definition, as the following example shows.

Let $L^2(\R)$ be the Hilbert space of square integrable functions and the $\chi_X$ characteristic function of a Borel set $X \in \hB(\R)$. We define the (one-dimensional) position measurement $Q:\hB(\R)\to \lh$ by $(Q(X)\fii)(x)=\chi_X(x)\fii(x)$ for all $x \in \R$, $X \in \hB(\R)$ and $\fii \in L^2(\R)$. This is especially a commutative, projective measurement.
\begin{example}\label{exposition}
    The position measurement $Q$ is enough to confirm the contextuality according to Definition \ref{defnormal}.
\end{example}
\begin{proof}
    Suppose $Q$ is contextuality non-confirming according to Definition \ref{defnormal}. In other words, all effects of $Q$ are fixed points of an EBC: 
    \begin{align}
        Q(X)=\Lambda_{EB}^*(Q(X)) \quad\forall X \in \hB(\R)
    \end{align}
    Now $Q$ is especially a continuous, extremal measurement.
    Thus by Theorem 1 of \cite{Jokinen_2024} this would imply that $L^2(\R)$ is a non-separable Hilbert space. This is a contradiction.
\end{proof}

To contrast this with existing results in finite-dimensional systems, we note that Example~\ref{exposition} is in similar spirit to a former work showing that a single measurement can lead to a proof of contextuality \cite{selby21}. The former work relies on a single, but nevertheless a non-commuting generalised measurement. Such measurement can be dilated to a commuting one, but proving contextuality from the dilation will require additional assumptions to the original definition of noncontextuality given in Ref.~\cite{spekkens05}. Technically speaking, such extra assumptions rise from treating two bipartite preparations on the dilation system as operationally equivalent, if their marginal states on one side are equal. This can be justified in scenarios, where an additional system is irrelevant to one's theory, e.g. in case the additional system expands the quantum state to include a note on the computer about which state was prepared. However, when probing which measurements can lead to contextuality, such assumption is considerably limiting the power of a noncontextual model. In contrast, Example~\ref{exposition} shows contextuality from commutativity without additional assumptions.

It would therefore seem that a classical, commutative, measurement is contextuality-confirming. This suggests that a direct application of noncontextual models to continuous variable systems is problematic. In the following, we propose an alternative approximative definition, that coincides with the standard definition for discrete systems. Here, $E(\hM)$ denotes the effects of the set of measurements $\hM$ i.e. the union of the ranges of all POVMs in $\hM$.
\begin{definition}\label{defapprox}
    (\textbf{Approximately contextuality non-confirming sets}) \\
     Let $\hM$ be a set of measurements and $\hS$ a set of states. Then $\hM$ is approximately contextuality non-confirming if there exists a POVM $G:\hB(\Omega) \to \lh$, defined in a standard Borel space, such that for every $\varepsilon>0$ and for every finite subset $\hF \subset E(\hM)$ there is an EB-channel $\Lambda_{\hF,G}^{*}:\lh \to \lh$ of the form
    \begin{align}\label{eqapproxeb}
        \Lambda_{\hF,G}^{*}(A)=\sum_i \tr{\sigma_i^\F A}G(X_i^\hF)
    \end{align}
    such that for all $\rho \in \sh$ and $E \in \hF$ the following holds.
    \begin{align}\label{eqapproxdeft}
        |\tr{\rho(\Lambda_{\hF,G}^{*}(E)-E) }|<\varepsilon
    \end{align}
    Similarly, $\hS$ is approximately contextuality non-confirming if equation \eqref{eqapproxdeft} holds for all $\rho \in \hS$ and all finite sets of effects $\hF \subset \lh$. 
\end{definition}

It is straightforward to verify that Example \ref{exposition} does not hold with Definition \ref{defapprox}; in fact, any finite set of elements of a single PVM are (even exact) fixed points of some EBC of the form \eqref{eqapproxeb}. The details are given later in Subsection "Sharp measurements and the non-normal definition".

We will show in Section V that the approximate definition can be equivalently cast as a modification of Definition~\ref{defnormal}, where the measurement response functions are weakened by allowing them to be non-constructive, i.e. non-normal. More formally, we have the following.

\begin{definition}\label{defmeasresp}
    Let $G:\hB(\Omega) \to \lh$ be a POVM and $\rho_0 \in \sh$ a faithful state, meaning that if $\tr{\rho_0A}=0$ for $A\ge 0$ then $A=0$. A $G$-\emph{measurement response} is a positive linear map $T:\lh \to L^\infty(\Omega,\tr{\rho_0G)})$ such that $T(I)=1$, where $1$ is the unit of $L^\infty(\Omega,\tr{\rho_0G}).$
\end{definition}
Note that this definition essentially does not depend on the chosen state $\rho_0$, since the measures $\tr{\rho G}$ are all mutually absolutely continuous when $\rho$ is a faithful state. \par 
We now state the definition of generalized contextuality using these measurement responses. 
\begin{definition}\label{defnonnormal}
    (\textbf{Non-normal contextuality non-confirming sets}) \\
    Let $\hM\subset \hO$ be a set of measurements and $\hS \subset \sh$ a set of states. $\hM$ is contextuality non-confirming if there is a POVM $G:\hB(\Omega) \to \lh$, a faithful state $\rho_0 \in \sh$ and a $G$-measurement response $T:\lh \to L^\infty(\Omega,\tr{\rho_0G)}$ such that we have the following for all $M\in \hM$ and $\rho\in \sh$:
    \begin{align}\label{eqdefnonnormalont}
        \tr{\rho M(X)}=\int_\Omega T(M(X))(\lambda) \, d(\tr{\rho G(\lambda)})
    \end{align}
    for all $X\in \hB(\Omega_M)$.
    Similarily $\hS$ contextuality non-confirming if equation \eqref{eqdefnonnormalont} holds for all $\rho \in \hS$ and $M \in \hO$. 
\end{definition}

In the later sections we examine the predictions of these definitions and the interplay between them. An informal summary of the relationships between the three definitions is given in Figure \ref{figcompare}.

\begin{figure}
\includegraphics[width=\linewidth]{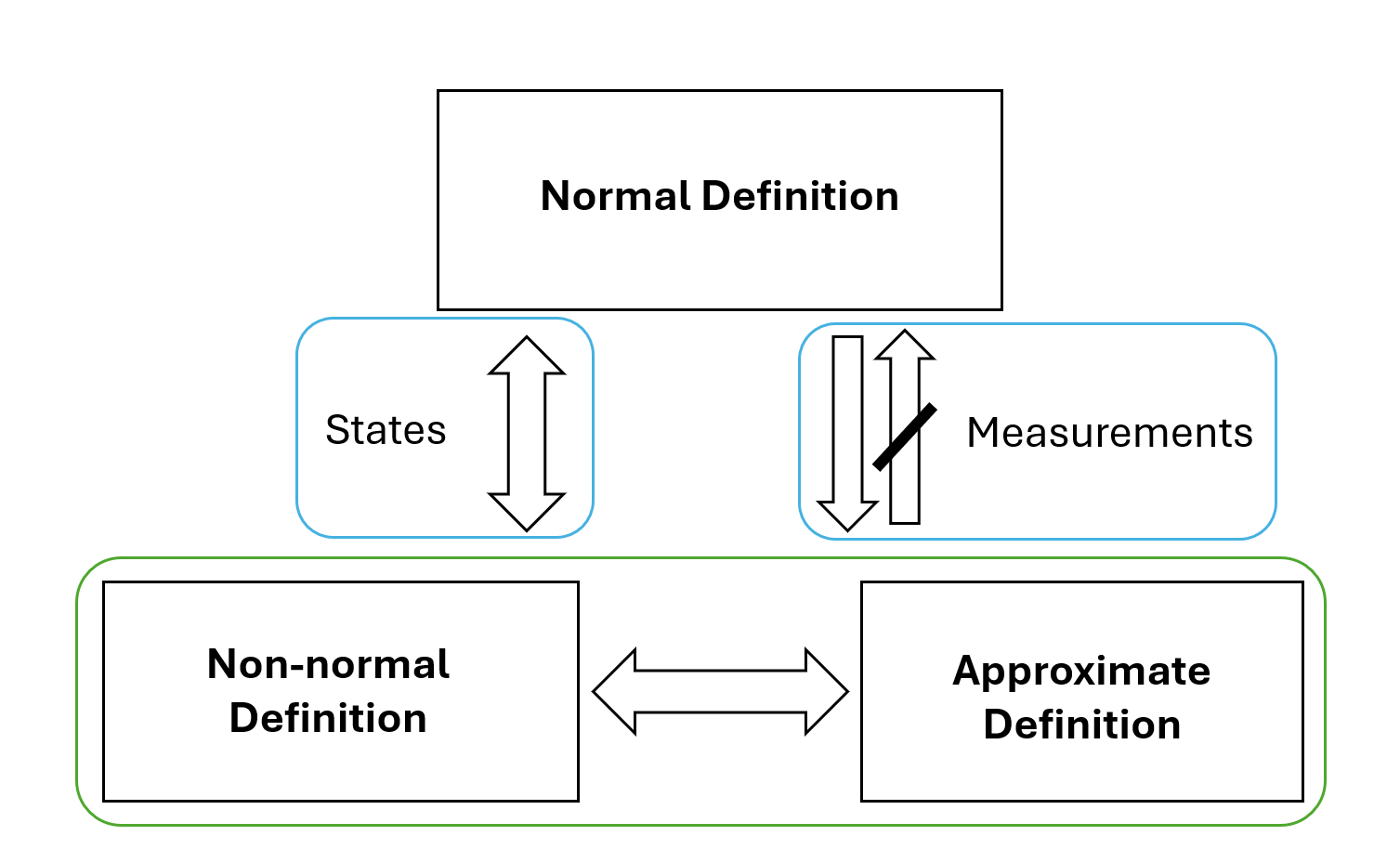}
\caption{Comparison of the three definitions of contextuality non-confirming sets. The non-normal definition and the approximate definition are equivalent (Proposition \ref{propapproxnonnormalequiv}), outlined in green. Furthermore all three definitions are equivalent when considering contextuality non-confirming states (Proposition \ref{propdefequivalence}). Finally for contextuality non-confirming sets of measurements Definitions \ref{defnonnormal} and \ref{defapprox} are strictly weaker than Definition \ref{defnormal} (Example \ref{exposition}, Corollary \ref{corpvmdiscrete} and Proposition \ref{proppvmapproximate}). The difference between the definitions arises from the continuity choices of the measurement responses.  }
\label{figcompare}
\end{figure}

\subsection{The probabilistic interpretation of measurement response functions}

The form of the measurement response functions is linked to their probability nature. In this subsection we briefly examine this in the context of the original work on generalized contextuality \cite{spekkens05}. 

In \cite{spekkens05} measurement response functions $p(\cdot|\lambda)$ should give \emph{probability distributions} when combined with POVMs. In the following, we show that if one wishes to follow this notion and interpret the measurement responses as probability distributions, one must
use Definition~\ref{defnormal}.

Let us first explain this in the case when the measurement responses are defined pointwise. 
In this case, if our Hilbert space is infinite-dimensional, the convex linearity of $p(\cdot|\lambda)$ does not in general give us similar structure as in the finite-dimensional case Eq.~\eqref{eqontmodelquant}. Indeed, extending $p(\cdot|\lambda)$ linearly from the effects of $\hO$ to all operators in $\lh$ gives that $p(\cdot|\lambda)$ belongs to the $C^*$-algebraic states of $\lh$ given by
\begin{align}
    &\hS(\lh):=\{\omega \in \lh^* \; | \;  \omega\ge 0, \omega(I)=1 \}.
\end{align}
Here $\lh^*$ is the Banach-space dual of $\lh$ with respect to the sup-norm $\Vert \cdot \Vert$. \par 

For infinite-dimensional Hilbert spaces, $\sh \subsetneq \hS(\lh)$ (interpreting $\sh$ as dual elements $A \mapsto \tr{\rho A}$). The elements of $\sh$ are simply density operators. However, the set of $C^*$-algebraic states $\hS(\lh)$ is more delicate. Let us elaborate shortly on this notion. \emph{Singular states} are examples of states that are not representable by a density operator \cite{Takesaki1979,Blackadar2006_2}. A singular state $\omega \in \hS(\lh)$ is a non-constructive function such that for all finite rank operators $F \in \lh$ we have $\omega(F)=0$. Equivalently $\omega(K)=0$ holds for all compact operators $K \in \hK(\hi)$, which is the $\Vert \cdot \Vert$-closure of the finite-rank operators. A feature of singular states is that they are not continuous in the \emph{ultraweak topology} defined as the coarsest topology such that the "expectation value maps" $\lh \ni A \mapsto \tr{T A}$ are continuous for all trace-class operators $T \in \th$. As a reminder, the trace-class operators consist of the operators $T$ for which the trace converges: $\tr{|T|}<\infty$.  Of course there are many other relevant topologies (e.g. weak, strong, ultrastrong) in which the singular states are not continuous. In fact, the ultraweakly continuous states are exactly the \emph{normal} states in $\hS(\lh)$, which are further equivalent to states given by $\sh$ \cite{Blackadar2006_2}. A normal state $\phi$ is such that for every increasing net $\{A_\alpha\} \subset \lh_+$ the identity $\phi(\sup_\alpha A_\alpha)=\sup_\alpha \phi(A_\alpha)$ holds. We shall henceforth call any state that is not normal a \emph{non-normal} state. Furthermore in Definition \ref{defnormal} the measurement response is chosen to be represented with a density operator i.e. normal and hence we call it the normal definition.

The normality of a state $\omega \in \hS(\lh)$ is intimately linked to its probability nature: the map $X \mapsto \omega(M(X))$ is a probability measure for all POVMs $M \in \hO$ exactly when $\omega$ is a normal state. The crux of the problem is $\sigma$-additivity, which can be seen with the following standard result. 
\begin{proposition}\label{propsigmaadd}
    A state $\omega \in \hS(\lh)$ is normal if and only if for all POVMs $M \in \hO$ the map $X \mapsto \omega(M(X))$ is $\sigma$-additive.
\end{proposition}

The trivial direction can be proven straight from definition: If $\omega$ is normal, then $X \mapsto \omega(M(X))$ is a probability measure by the definition of a POVM, thus especially $\sigma$-additive. The other direction is a straightforward consequence of Theorem III.2.1.4 (viii) in \cite{Blackadar2006_2}.

Thus Proposition \ref{propsigmaadd} implies that convex linearity for the measurement responses $p(\cdot|\lambda)$ alone is not enough to guarantee they are valid ($\sigma$-additive) probabilities. In the original work on generalized contextuality \cite{spekkens05}, the measurement responses are pointwise defined, and probability distributions. This necessarily leads to Definition \ref{defnormal} in the context of full quantum theory according to Proposition \ref{propsigmaadd}. \par 
Furthermore, even if the measurement response is not necessarily defined pointwise (i.e. we have a general $G$-measurement response according to Definition \ref{defmeasresp}), we come to the same conclusion. Indeed, if the $G$-measurement response $T$ produces valid probabilities i.e. $X \mapsto \int_\Omega f(\lambda)T(M(X)) \, d\tr{\rho_0 G(\lambda)}$ is a probability measure for all $f\ge 0$ s.t. $\int_\Omega f(\lambda) \,d\tr{\rho_0 G(\lambda)}=1 $, then $T$ is (by definition) continuous between the weak$^*$ topologies of $\lh$ and $L^\infty(\Omega,\tr{\rho_0 G})$. This leads to the form of Definition \ref{defnormal}, cf. Proposition \ref{propdefnormalreduction} . As a reminder, the weak$^*$-topology of $L^\infty$ is the initial topology induced by the set of integrable functions $L^1$.

Thus if one wishes to interpret the measurement responses as probability distributions, one must necessarily use Definition \ref{defnormal}. However, this would also require accepting that the position measurement alone is contextuality-confirming. This motivates one to consider an approximate definition such as Definition \ref{defapprox}, thus keeping the probabilistic interpretation on every successive approximation, while not allowing a commuting measurement to confirm contextuality.

We note that the need for modifying Definition \ref{defnormal} comes exclusively from measurements. On the state side Definition \ref{defnormal} does not lead to commuting contextuality confirming sets. This, and further implications on contextuality non-confirming sets with respect to Definition \ref{defnormal} are explored in the next section.

\section{IV. Normal ontological models}

Intuitively, normal noncontextual models, cf. Eq.~(\ref{eqontmodelquant}) ask whether correlations between the sent state and the outcome of measurements can be explained by sending a classical message $\lambda$ and preparing a state $\sigma_\lambda$ based on this information, i.e. no quantum information is passed from the preparation to the measurements. Classical information can be copied, i.e. broadcast, meaning that noncontextual models are related to broadcasting. In Ref.~\cite{jokinen2024nobroadcasting}, a more formal equivalence between noncontextual models and the no-broadcasting theorem was shown using fixed points of an entanglement breaking channels. Here, we aim at characterising contextuality confirming sets of states and measurements in terms of in the infinite-dimensional setting. For states, we show a tight connection to broadcasting, whereas for measurements we prove that contextuality non-confirming in the sense of Definition~\ref{defnormal} implies broadcastability, which we further characterise by a post-processing relation. Broadcastability is formally defined as follows.
\begin{definition}
    Let $\hS \subset \sh$, $\hA,\hB \subset \hO$. Then the triple $(\hS,\hA,\hB)$ is broadcastable if and only if there exists a quantum channel $\Lambda:\th \to \hT(\hi \otimes \hi)$ such that for all $\rho \in \hS$,$A \in \hA$, $B \in \hB$ the following equalities hold for all $X \in \hB(\Omega_A),Y \in \hB(\Omega_B)$:
    \begin{align}\label{broadcastingcondition2}
        \tr{A(X)\mathrm{tr}_{2}{\Lambda(\rho)}}&=\tr{A(X) \rho}, \\
        \tr{B(Y)\mathrm{tr}_{1}{\Lambda(\rho)}}&=\tr{B(Y) \rho}.
    \end{align}
\end{definition}
In this paper we are interested only in two types of broadcasting, i.e. broadcasting some states $\hS$ with all measurements or broadcasting some measurements with all states.  Furthermore both parties have the same measurements. This captures the notion of broadcasting a set of states or a set of measurements respectively. Thus for the rest of the article we use the following shorter notation: 
\begin{itemize}
    \item $\hS$ broadcastable $\leftrightarrow$ $(\hS,\hO,\hO)$ broadcastable
    \item $\hM$ broadcastable $\leftrightarrow$ $(\sh,\hM,\hM)$ broadcastable.
\end{itemize}
We now examine the connection between contextuality and broadcasting for both quantum states and measurements. The next subsection contains the results relating normal ontological model to broadcasting and to other notions of classicality. The technical structure and details are postponed to the section after. For the rest of the article, given a quantum channel $\Lambda:\th \to \th$, we use the notation $\fix{\Lambda}$ and $\fix{\Lambda^*}$ for the following fixed point sets.
\begin{align}
    \fix{\Lambda}&=\{T \in \th \, | \, \Lambda(T)=T\} \\
    \fix{\Lambda^*}&=\{A \in \lh \, | \, \Lambda^*(A)=A\}
\end{align}

\subsection{Contextuality and broadcasting for quantum states}
We begin by examining states, where we find results analogous to the finite-dimensional case \cite{jokinen2024nobroadcasting}. The following formal statement is proven in the Section "Algebraic structure of broadcasting" and we note that for states there is no difference between the normal and the non-normal definitions, cf. Proposition \ref{propdefequivalence}.

\begin{theorem}\label{thmstatecont}
    Let $\hS \subset \sh$ be a set of states. Then the following are equivalent. 
    \begin{enumerate}
        \item $\hS$ is contextuality non-confirming. \label{c_ont}
        \item $\hS \subset \fix{\Lambda_{EB}}$ for an EB-channel $\Lambda_{EB}$. \label{c_eb}
        \item $\hS$ is broadcastable. \label{c_br}
        \item $\hS$ is commutative.  \label{c_com}
    \end{enumerate}
\end{theorem}

According to Theorem \ref{thmstatecont} (\ref{c_com}.) one also sees that contextuality confirmation and nonzero set coherence \cite{designolle21b} are equivalent for the state set $\hS$. \par 
There is also another interesting equivalent condition to $\hS$ being contextuality non-confirming. For this we introduce the \emph{Koashi-Imoto decomposition} of a state \cite{Koashi2002,Galke2024}. Let $\hS$ be a set of states and $P_\hS$ the least upper bound of supports of $\hS$. Then the Hilbert space $P_\hS(\hi)$ decomposes as $P_\hS(\hi)=\bigoplus_j \hJ_j \otimes \ki_j$ and the states $\rho \in \hS$ decompose in this Hilbert space structure as
\begin{align}
   \rho= \bigoplus_j p_j(\rho)\sigma_j(\rho) \otimes_j\om_j,
\end{align}
where $p_j(\rho)$ is a probability distribution, $\sigma_j(\rho)$ and $\om_j$ states. This is called the Koashi-Imoto decomposition. Its existence and some additional properties are proven in \cite{Koashi2002}. With the results of \cite{Galke2024} we then get the following additional characterization.
\begin{corollary}
    Let $\hS \subset \sh$ be a set of states and $P_\hS=\bigvee_{\rho \in \hS} \mathrm{supp}(\rho)$. Then $\hS$ is contextuality non-confirming if and only if $\dim \hJ_j=1$ for all $j$ in the Koashi-Imoto decomposition $P_\hS(\hi)=\bigoplus_j \hJ_j \otimes \ki_j$.
\end{corollary}
\begin{proof}
    Proposition 6 in \cite{Galke2024}.
\end{proof}
The shift-channel  (cf. ref. \cite{salzmann2024robustnessfixedpointsquantum} or equation \eqref{eqshiftchannel}) and Example \ref{exqfunction} show that there exist EB-channels that do not give noncontextual explanations to any states in infinite dimensional systems. This is unlike in the finite-dimensional systems, where the set of quantum states is compact in trace-norm, and thus e.g. the Schauder-Tychonoff fixed point theorem \cite[V.10.5]{MR117523} guarantees a fixed point. The set of density operators is not compact in the trace-norm when the dimension is infinite, but by the Banach-Alaoglu theorem \cite[3.15]{rudin1991functional} the state space $\hS(\lh)$ is compact in the weak$^*$-topology. Thus again by the Schauder-Tychonoff fixed point theorem (and the normality of a quantum channel), there is a "fixed point state" for an EB-channel, but it does not have to be representable by a density operator.

\subsection{Contextuality and broadcasting for quantum measurements}

Although Definition~\ref{defnormal} leads to commuting contextuality confirming measurements, we explore further its merit as a notion of classicality through broadcasting. It is shown in Proposition \ref{propwstar} and Corollary \ref{corollaryeb} that the fixed points of the marginal of a broadcasting channel, thus especially the fixed points of an EB-channel, form a commutative von Neumann algebra. It is known that such algebras are always $*-$isomorphic to a classical function space $L^\infty(\Omega,\mu)$ of (equivalence classes of) essentially bounded functions on $\Omega$ with respect to the measure $\mu$ \cite[Proposition 1.18.1]{Sakai1998}. Now all effects of broadcastable POVMs are within this fixed point set, so we can boil this down to the following (a bit imprecisely formalised) observation.
\begin{observation}
    A set of (normal) contextuality non-confirming measurements can be embedded into a classical function space $L^\infty(\Omega,\mu)$. 
\end{observation}
In the following we shall see that the implication in the above Observation can not be reversed. However, for the non-normal definition, we show later that such embeddability in the case of an algebra on a separable Hilbert space is sufficient for the inverse implication.

To elaborate on further restrictions that measurements obeying a noncontextual ontological model must fulfill, we present the following results, the proofs of which are in the Section "Algebraic structure of broadcasting". Recall that the effects of the set of POVMs $\hM$ is denoted by $E(\hM)$. 
\begin{theorem}\label{thmmeaschar}
    Let $\hM \subset \hO$ be a set of POVMs. Then for the following statements
    \begin{enumerate}
        \item $\hM$ is (normal) contextuality non-confirming. \label{d_ont}
        \item $E(\hM) \subset \fix{\Lambda_{EB}^*}$ for some EB-channel $\Lambda_{EB}^*$. \label{d_eb}
        \item $\hM$ is broadcastable. \label{d_br}
        \item All $M \in \hM$ decompose into broadcastable "normal" and "singular" subnormalized POVMs: $M=M_n\oplus M_s$, and all $M_n$ can be post-processed from a single discrete eigenvalue-1 POVM. \label{d_dec}
    \end{enumerate}
    the implications $\eqref{d_ont} \Leftrightarrow \eqref{d_eb} \Rightarrow \eqref{d_br} \Leftrightarrow \eqref{d_dec}$ hold. Furthermore, if there is a POVM $M \in \hM$ with no singular part, $M_s=0$ , then all of these statements are equivalent.
\end{theorem}

The point \ref{d_dec}. in the previous theorem thus implies that the contextuality non-confirmation property enforces a type of discreteness on the POVMs of $\hM$, i.e., the normal parts can be post-processed from a discrete POVM. This is especially apparent if we consider only sharp measurements, which is illustrated in the next corollary. 
\begin{corollary}\label{corpvmdiscrete}
    If $\hM$ is a set of PVMs, then all the conditions of Theorem \ref{thmmeaschar} are equivalent. Furthermore all $P$ satisfying these conditions are discrete and mutually commutative. 
\end{corollary}

Note that in the preceding corollary if the measure spaces are assumed to be more general (i.e. not all atoms are singlets), "discrete" can just be replaced with "purely atomic". A \emph{purely atomic} POVM $E: \hB(\Omega_E) \to \lh$ here is a POVM such that $\Omega_E$ is a countable union of \emph{atoms} of $E$, where an atom is a set $A \in \hB(\Omega_E)$ with $E(A)\neq0$ and such that\ $B \subset A$ implies $E(B)=0$ or $E(B)=E(A)$ for measurable $B$. We emphasize here that the separability of the Hilbert space $\hi$ is essential in Corollary \ref{corpvmdiscrete}. This is because the discreteness of a measure requires a \emph{countable} decomposition of the outcome space to atoms. If $\hi$ would be allowed to be nonseparable, then there could be an uncountable amount of minimal projections (atoms) in the von Neumann algebra generated by a PVM, thus not guaranteeing pure atomicity.  \par 
Corollary \ref{corpvmdiscrete} leads to an interesting conclusion that only appears when considering continuous-variable measurements. Indeed, this result forces contextuality non-confirming sharp measurements to be embeddable to the classical broadcasting algebra \emph{and} discrete. Taking the negation of this, we get the following (a bit imprecisely formalised) observation.
\begin{observation}
    Sharp measurements are (normal) contextuality confirming if and only if they are not embeddable to a classical function space \textbf{or} not discrete.
\end{observation}

This implies that there can exist classically embeddable measurements proving the contextuality of quantum theory, cf. Example \ref{exposition}. However, if one wants to find a definition of contextuality in which measurements embeddable to classical space are contextuality non-confirming, then one could expand Definition \ref{defnormal}, as proposed before. This motivates the relaxing of Definition \ref{defnormal} to the realm outside of regular probability theory.

\subsection{Algebraic structure of broadcasting}\label{sectiontechnical}
In order to characterize the normal contextuality non-confirming sets, we need to find out general facts about the fixed points of EB-channels. It turns out that the inherent symmetric broadcasting nature of EB-channels plays a significant role in the structure of the fixed points. Indeed, an EB-channel $\Lambda^*(A)=\int_\Omega\tr{\sigma_\lambda A} \, dG(\lambda)$ forms a symmetric (broadcasting) channel $\Phi^*:\hL(\hi \otimes \hi) \to \lh$ defined as follows:
\begin{align}\label{eqsymmetriceb}
    \Phi^*(A):=\int_\Omega \tr{(\sigma_\lambda \otimes \sigma_\lambda)A} \, dG(\lambda).
\end{align}
The channel $\Phi^*$ clearly is symmetric in simple tensors i.e. 
\begin{align}\label{eqsymmetry}
\Phi^*(A \otimes B)=\Phi^*(B \otimes A)
\end{align}
for all $A,B \in \lh$. We shall call any channel obeying equation  \eqref{eqsymmetry} for all $A,B \in \lh$ a \emph{symmetric channel}. Furthermore the broadcastable POVMs of $\Phi^*$ are exactly the fixed points of the EB-channel $\Lambda^*$, defined as $\Lambda^*(A)=\Phi^*(A \otimes I)=\Phi^*(I \otimes A)$ for all $A \in \lh$. We thus focus now a bit more generally on the fixed points of the marginals of a general symmetric channel. Let us denote the marginal of a symmetric channel $\Phi^*:\hL(\hi \otimes \hi) \to \lh$ by $\Phi_0^*:\lh \to \lh$. In other words for all $A \in \lh$, $\Phi_0^*(A):=\Phi^*(A \otimes I)=\Phi^*(I \otimes A)$. Note that broadcasting with a channel $\Theta^*:\hL(\hi \otimes \hi)\to \lh$ can always be done with a symmetric channel preserving the fixed points of the marginal. In other words one can define the channel $\Theta_{sym}^*=\frac{1}{2}(\Theta^*+\Theta^* \circ V_{swap})$, where $V_{swap}(A \otimes B)=B \otimes A$ is the unitary swap-channel. This motivates us to restrict our considerations to such channels in the following.\par
We now show a key structural result on the fixed points of $\Phi_0^*$: The set $\fix{\Phi_0^*}$ can be made into a \emph{commutative abstract von Neumann algebra} using broadcasting. An abstract von Neumann algebra $M$ is a $C^*$-algebra that has a predual, i.e., there exists a Banach space $M_*$ such that the dual $(M_*)^*$ is isometrically isomorphic to $M$. Every abstract von Neumann algebra has a normal $*$-isomorphism to some \emph{concrete von Neumann algebra} $N \subset \kh$ \cite{Sakai1998}. A concrete von Neumann algebra is a von Neumann subalgebra of some $\kh$. Here $\ki$ can be non-separable. For the rest of the paper, we will call abstract von Neumann algebras just von Neumann algebras for brevity. \par We begin by presenting the following lemma. For linear maps $\hE:\lh \to \lh$ we refer to the initial topology generated by $\hE \mapsto \tr{T\hE(A)}$ for all $A \in \lh, T \in \th$ as the point-ultraweak topology.
\begin{lemma}\label{lemmafixproj}
    Let $\Lambda^*:\lh \to \lh$ be a quantum channel. Then there exists an idempotent completely positive map $\psi_0:\lh \to \lh$ such that $\psi_0 \circ \Lambda^*=\psi_0$ and $\psi_0(\lh)=\fix{\Lambda^*}$ given as a limit of a point-ultraweak subnet of the sequence $\left( \frac{1}{n}\sum_{k=0}^{n-1} (\Lambda^*)^k\right)_{n \in \N} $.
\end{lemma}
\begin{proof}
    See for example \cite[Theorem 2.4]{Arias2002}.
\end{proof}

\begin{proposition}\label{propwstar}
    Let $\hi$ be a separable Hilbert space and $\Phi^*:\hL(\hi \otimes \hi) \to \lh$ a symmetric channel. Then for $\Phi^*_0:=\Phi^*( \cdot \otimes I)=\Phi^*(I \otimes \cdot )$,  there is a product defined on $\fix{\Phi_0^*}$ which makes it a commutative von Neumann algebra with respect to the operator norm and the involution inherited from $\lh$. 
\end{proposition}
\begin{proof}
   The idea is to use Proposition 3 of \cite{kuramochi20}.  We first show that $\fix{\Phi_0^*}$ is a Banach space. But this is clear since the set of fixed points of any bounded linear map $\Xi:\lh \to \lh$ is the continuous preimage of a closed set: $\fix{\Xi}=(\Xi-\mathrm{id})^{-1}(\{0\})$. \par 
    Next we define the product in this space. Now since $\Phi^*$ is normal, so is $\Phi_0^*$. Let then $\psi_0:\lh \to \lh$ be the map given by Lemma \ref{lemmafixproj}. With this we define the bilinear map $\mathsf{B}:\fix{\Phi_0^*} \times \fix{\Phi_0^*} \to \fix{\Phi_0^*} $ by 
    \begin{align}
        \mathsf{B}(A,B):=\psi_0(\Phi^*(A \otimes B))
    \end{align}
    for all $(A,B) \in \fix{\Phi_0^*} \times \fix{\Phi_0^*} $.  Now if $A,B\geq 0$, then $\mathsf{B}(A,B)\geq 0$, because $\psi_0$ and $\Phi^*$ are completely positive. Furthermore for any $A \in \fix{\Phi_0^*}$  we have
    \begin{align}
        \mathsf{B}(A,I)=\psi_0(\Phi^*(A \otimes I))=A=\psi_0(\Phi^*(I \otimes A))=\mathsf{B}(I,A).
    \end{align}
    In fact $\mathsf{B}$ is a bilinear map fulfilling the conditions of Proposition 3 in \cite{kuramochi20} and hence $\mathsf{B}$ is the unique associative bipositive product on $\fix{\Phi_0^*}$.  \par 
    Thus we can now define the product $$A \bullet B=\mathsf{B}(A,B)=\psi_0(\Phi^*(A \otimes B))$$ for all $A,B \in \fix{\Phi_0^*}$. Furthermore we see that $\fix{\Phi_0^*}$ is a commutative unital $C^*$-algebra with respect to the product $\bullet$ with the operator norm and involution inherited from $\lh$. This is indeed a proper involution in the product $\bullet$, since the (completely) positive maps $\psi_0$, $\Phi^*$ preserve the adjoint. The Banach algebra condition is easy to see: $\Vert A \bullet B \Vert = \Vert \psi_0(\Phi^*(A \otimes B))\Vert \leq \Vert A \otimes B \Vert =\Vert A \Vert \Vert B \Vert$. The $C^*$-condition $\Vert A^* \bullet A \Vert =\Vert A\Vert^2 $ can be seen using the fact that for all extreme states $\phi$ in $\fix{\Phi_0^*}^*$ we have that for all $A,B \in \fix{\Phi_0^*}$
    \begin{align}
        \phi(\mathsf{B}(A,B))=\phi(A)\phi(B) \label{eqcharacter}
    \end{align}
    This is noted in the proof of Proposition 3 in  \cite{kuramochi20}, but we will give the proof of equation \eqref{eqcharacter} for the readers convenience. For this, let $\phi$ be an extreme state and $0\leq E \leq 1$, $E \in \fix{\Phi_0^*}$ be such that $0<\phi(E)<1$. Then for all $A \in \fix{\Phi_0^*}$ we have that
    \begin{align}
        \phi(A)=\phi(E)\frac{\phi(A \bullet E)}{\phi(E)}+\phi( E^\perp)\frac{\phi(A \bullet E^\perp)}{\phi( E^\perp )}
    \end{align}
    Thus by extremality we have that $\phi(A \bullet E)=\phi(A)\phi(E)$. Furthermore if $\phi(E)=0$, then $\phi(A)=\phi(A \bullet E^\perp)$, which implies $\phi(A \bullet E)=0=\phi(A)\phi(E) $. If $\phi(E)=1$, then $\phi(A \bullet E)=\phi(A)=\phi(A)\phi(E)$. Thus for all effects $E$, equation \eqref{eqcharacter} holds. Since this holds for all effects, it holds for all elements in $\fix{\Phi_0^*}$. \par  
    Using this condition we see that $\phi(A^* \bullet A)=|\phi(A)|^2$, from which the $C^*$-condition follows by optimizing over all extreme states. The other conditions for $\fix{\Phi_0^*}$ to be a $C^*$-algebra are trivial. We also note that the positive cone of $\fix{\Phi_0^*}$ (in the usual order of $\lh$) is equal to the $C^*$-algebra positive cone. In other words we want to show that $A \in \fix{\Phi_0^*}, A\geq 0 \Leftrightarrow A=B^* \bullet B$ for some $B \in \fix{\Phi_0^*}$. The "($\Leftarrow$)" is easily seen with equation \eqref{eqcharacter}. The "($\Rightarrow$)"-direction can be seen for example by using the isomorphism $\fix{\Phi_0^*}_s \cong C(\Omega)$ for some compact Hausdorff space $\Omega$ derived in \cite{kuramochi20}. Here  $\fix{\Phi_0^*}_s $ is the self-adjoint elements of the fixed points. Let $\Psi$ be the isomorphism in the proof of Proposition 3 in \cite{kuramochi20}. In the context of said Proposition, $\Psi$ is an order isomorphism, but from the proof one can easily see that $\Psi$ is actually also an algebraic isomorphism with respect to the product $\bullet$.  Now $A=\Psi^{-1}(f_A)$ for a continuous positive function $f_A$. Thus $\sqrt{f_A}$ is a positive continuous function, and $\Psi^{-1}(\sqrt{f_A}) \bullet \Psi^{-1}(\sqrt{f_A})=\Psi^{-1}(f_A)=A$. This proves the equivalence between the positive cones. \par 
    Therefore $\fix{\Phi_0^*}$ is a commutative unital $C^*$-algebra with the product $\bullet$.
    Furthermore, since $\Phi_0^*$ is normal, $\fix{\Phi_0^*}$ is an ultraweakly closed subspace in $\lh$. Then for example Proposition 2.12 of Chapter V in \cite{198897} implies that $\fix{\Phi_0^*}$ has a predual $\fix{\Phi_0^*}_*$ and is thus a von Neumann algebra. 
    \end{proof}
    For brevity, we shall call the algebra $(\fix{\Phi_0^*},\psi_0(\Phi^*(\cdot \otimes \cdot)))$ the \emph{broadcasting algebra}. \par 
    Since EB-channels define the symmetric channel \eqref{eqsymmetriceb}, we get the the following.
\begin{corollary}\label{corollaryeb}
    The (Heisenberg picture) fixed points of an entanglement breaking channel form a commutative von Neumann algebra. 
\end{corollary}
We emphasize that the product $\fix{\Phi_0^*} \times \fix{\Phi_0^*} \ni (A,B)\mapsto \psi_0(\Phi^*(A \otimes B))$ of the broadcasting algebra is not in general equal to the usual operator product of $\lh$. The product can however be connected to the usual product of $\lh$ in a "quotient" sense. The existence of the completely positive idempotent map $\psi_0: \lh \to \lh$ with $\psi_0(\lh)=\fix{\Phi_0^*}$ makes the fixed points an \emph{injective operator system} \cite{CHOI1977156}. It was shown by Choi and Effros that every injective operator system has a natural $C^*$-algebra structure with the product $(A,B) \mapsto \psi_0(AB)$ \cite{CHOI1977156}, sometimes also called the Choi-Effros product. For the broadcasting algebra the Choi-Effros product turns out to be equal to the broadcasting product as the following proposition shows. For this proposition, let us introduce a few concepts needed for the proof. Here a \emph{projection in $\bullet$} is an operator $P \in \fix{\Phi_0^*}$ such that $P \bullet P=\psi_0(\Phi^*(P \otimes P))=P$ and $P^*=P$. Note that projections are always effects: $0 \leq P \leq I$. \par 
Let then $\hA$ be a unital $C^*$-algebra and $\hE:\hA \mapsto \lh$ be a unital completely positive map. A \emph{minimal Stinespring dilation} is a triple $(V,\pi,\ki)$, where $\ki$ is a (possibly nonseparable) Hilbert space, $V:\hi \to \ki$ an isometry and $\pi:\hA \to \kh$ a $*$-representation, with the linear span of $\{\pi(A)V\fii \; | \; A \in \hA, \fii \in \hi\}$ dense in $\ki$.
\begin{proposition}\label{propproduct}
    Let $\Phi_0^*$ be the marginal of a symmetric channel $\Phi^*$ and $\psi_0$ as in Lemma \ref{lemmafixproj}. Then for all $A,B \in \fix{\Phi_0^*}$ the following holds.
    \begin{align}
        A \bullet B=\psi_0(\Phi^*(A \otimes B))=\psi_0(AB)
    \end{align}
\end{proposition}
\begin{proof}
    Let $C^*(\fix{\Phi_0^*})$ denote the $C^*$-algebra generated by the fixed points with the $\lh$-product and sup-norm. Furthermore, let $\psi$ be the restriction of $\psi_0$ to $C^*(\fix{\Phi_0^*})$. We begin this proof by showing that  $\ker \psi$ is an ideal of $C^*(\fix{\Phi_0^*})$. Let thus $K \in \ker \psi$ and $A\in \fix{\Phi_0^*}$. Now as an intermediate result in the proof of \cite[Theorem 3.1]{CHOI1977156} (applied to our case) it is shown that for all $A \in \lh, B \in \fix{\Phi_0^*}$ we have that $\psi_0(AB)=\psi_0(A\psi_0(B))$ and $\psi_0(BA)=\psi_0(\psi_0(B)A)$. We apply this to the case where $B=K$, so that  $\psi_0(AK)=\psi_0(A\psi_0(K))=0$ and similarily for multiplication from the left. Hence $AK,KA\in \ker \psi$. One then inductively sees that this holds when $A$ is an arbitrary linear combination of finite products of operators from $\fix{\Phi_0^*}$ Since $\psi_0$ is norm-continuous, $\psi(KA)=\psi(AK)=0$ holds for all $A \in C^*(\fix{\Phi_0^*})$. Thus $\ker \psi$ is a two-sided ideal of $C^*(\fix{\Phi_0^*})$.  Let then $(V,\pi,\ki)$ be a minimal Stinespring dilation of $\psi: C^*(\fix{\Phi_0^*}) \to \lh$. Since $\ker \psi$ is an ideal, we have that $\ker \pi= \ker \psi$, which is seen as follows. Let $K \in \ker \psi$, $\fii,\eta \in \hi$ and $A,B \in C^*(\fix{\Phi_0^*})$. Then
    \begin{align}
        \ip{\pi(A)V\fii}{\pi(K)|\pi(B)V\eta}&=\ip{\fii}{V^*\pi(A^*KB)V|\eta}\\
        &=\ip{\fii}{\psi(A^*KB)|\eta}=0
    \end{align}
    Since the dilation is minimal, this implies that $\ip{\psi}{\pi(K)|\psi'}=0$ for all $\psi,\psi'\in \mathcal K$ and hence $\pi(K)=0$. The other direction is trivial.  \par 
    Let then $P \in \fix{\Phi_0^*}$ be a projection in $\bullet$. Then, since by Lemma \ref{lemmafixproj} we have $\psi_0=\psi_0 \circ \Phi^*_0$, the following holds.
    \begin{align*}
        \psi(PP^\perp)&=\psi_0(PP^\perp)=\psi_0(\Phi^*(PP^\perp \otimes I)) \\
        &=\psi_0(\Phi^*(PP^\perp \otimes P))+\psi_0(\Phi^*(PP^\perp \otimes P^\perp))\\
        &\leq P^\perp\bullet P + P\bullet P^\perp\\
        &=0,
    \end{align*}
    since $PP^\perp \leq P,P^\perp$ and $P$ is a projection.  Thus $PP^\perp \in \ker \psi$. Therefore also $\pi(PP^\perp)=0$, so $\pi(P)$ is a projection in $\kh$. \par  Let then $Q,P \in \fix{\Phi_0^*}$ be arbitrary projections in $\bullet$. Then 
    \begin{align}
        \pi(P \bullet Q)=\pi(P)\pi(P \bullet Q)+\pi(P^\perp)\pi(P\bullet Q)
    \end{align}
    Now $\bullet$ is commutative, so $P \bullet Q$ is a projection and $P \bullet Q \leq P$. Thus $\pi(P^\perp)\pi(P\bullet Q)=0$. Furthermore
    \begin{align}
        \pi(P)\pi(Q)=\pi(P)\pi(P \bullet Q)+\pi(P)\pi(P^\perp \bullet Q)
    \end{align}
    By a similar argument $\pi(P)\pi(P^\perp \bullet Q)=0$. Thus for all projections $P,Q$ we have $\pi(P \bullet Q)=\pi(PQ)$. Now since all elements of a von Neumann algebra are norm-limits of linear combinations of projections \cite{kuramochi20}, we have that $\pi(A \bullet B)=\pi(AB)$ by the norm-continuity of $\pi$. Therefore finally we get the desired result.
    \begin{align}
        A \bullet B&=\psi(A \bullet B)=V^*\pi(A \bullet B)V \\
        &=V^*\pi(AB)V=\psi(AB)=\psi_0(AB)
    \end{align}
\end{proof}
Proposition \ref{propproduct} especially implies that if $\fix{\Phi_0^*}$ is an algebra in the usual product of $\lh$, then it has to also be commutative in the same product. Furthermore, it even has to be a \emph{completely atomic} von Neumann algebra in this case i.e. it is generated by a discrete PVM. This is shown in the next theorem.
\begin{theorem}\label{thmlhalgebra}
    Suppose $\Phi_0^*$ is the marginal of a symmetric channel $\Phi^*$. Then if $\fix{\Phi_0^*}$ contains a unital subalgebra $\hN$ of $\lh$, then the ultraweak closure of $\hN$ is a completely atomic (concrete) von Neumann algebra. 
\end{theorem}
\begin{proof}
     Now since $\hN$ is a subalgebra and $\fix{\Phi_0^*}$ is ultraweakly closed and preserves adjoint, the von Neumann algebra generated by $\hN$, i.e. the ultraweak closure $\overline{\hN}^{uw}$, is contained within $\fix{\Phi_0^*}$. Let then $Q,P \in \overline{\hN}^{uw}$ be projections. These now commute by Proposition \ref{propproduct}. Then 
    \begin{align}
        \Phi^*(P \otimes Q)&\leq PQ=\Phi^*(PQ \otimes I)\\
        &=\Phi^*(PQ \otimes P)+\Phi^*(PQ \otimes P^\perp) \\
        &\leq \Phi^*(Q \otimes P)+\Phi^*(P \otimes P^\perp)\\
        &=\Phi^*(Q \otimes P)
    \end{align}
    Here the first inequality and last equality use the fact that $\Phi^*(Q\otimes P)$ is a common lower bound for $P,Q$ so that it's less than the greatest lower bound $PQ$. Therefore $\Phi^*(P \otimes Q)=PQ$. This extends to all $A,B \in \overline{\hN}^{uw}$ through continuity i.e. $\Phi^*(A \otimes B)=AB$. Therefore, by Proposition 3.6 in \cite{Kaniowski2014} we have that $\overline{\hN}^{uw}$ is completely atomic.
\end{proof}
Theorem \ref{thmlhalgebra} thus especially implies that if $\fix{\Phi_0^*}$ is itself a subalgebra of $\lh$, then it is completely atomic.
\begin{corollary}\label{corfullrankfixstate}
    Suppose $\Phi_0^*$ is the marginal of a symmetric channel $\Phi^*$. If there is a faithful full rank fixed point state for $\Phi_0: \th \to \th$, then $\fix{\Phi_0^*}$ is a completely atomic von Neumann algebra.
\end{corollary}
\begin{proof}
    The existence of a faithful full rank invariant state implies that $\fix{\Phi^*_0}$ is an algebra, see for example Lemma 2.2 in \cite{Arias2002}. Thus we get the claim from Theorem \ref{thmlhalgebra}.
\end{proof}
The results above imply that in these special cases the broadcasting algebra is forced to be commutative in the $\lh$-product. These results are, of course, applicable to the finite-dimensional framework and therefore the noncommutative broadcastable sets in \cite{jokinen2024nobroadcasting} imply that there are fixed point sets that are not subalgebras of $\lh$. \par 
We now concentrate on this general case, where the fixed points are not necessarily subalgebras of $\lh$. 
Here we intend to show that all broadcastable POVMs i.e. POVMs in $\fix{\Phi_0^*}$ can be divided into "normal" and "singular" parts. The normal part is shown to belong to a completely atomic broadcasting von Neumann algebra, whereas the structure of the singular part remains an interesting open question.
\par

In the following, a map $\hE:\lh \to \lh$ is called \emph{singular} if $\hK(\hi)\subset \ker \hE$.
\begin{proposition}\label{propsingornorm}
    Let $\Phi_0^*$ be the marginal of a symmetric channel $\Phi^*$ and $\psi_0$ the map given in Lemma \ref{lemmafixproj}. Then
    \begin{align}
        \psi_0=\psi_n \oplus \psi_s,
    \end{align}
    where $\psi_n: \lh \to \lh$ is normal, completely positive and idempotent, and $\psi_s:\lh \to \lh$ is singular, completely positive and idempotent. Especially $\psi_n(I)$ and $\psi_s(I)$ are orthogonal projections in $\bullet$ so that $\fix{\Phi_0^*}$ decomposes to a "normal" and "singular" ideal: 
    \begin{align}
        \fix{\Phi_0^*}=\psi_n(I)\bullet \fix{\Phi_0^*}\oplus \psi_s(I)\bullet \fix{\Phi_0^*}. 
    \end{align}
\end{proposition}
\begin{proof} 
Let $(V,\pi,\ki)$ be a minimal Stinespring dilation of $\psi_0$. Now by \cite[Theorem 2.14]{Takesaki1979} the representation $\pi:\lh \to \kh$ can be divided to normal and singular representations:
\begin{align}
    \pi=\pi_n \oplus \pi_s
\end{align}
Thus we may define $\psi_{n/s}=V^*\pi_{n/s}(\bullet)V$ so that 
\begin{align}
    \psi_0=\psi_n+\psi_s.
\end{align}
Obviously both $\psi_n$ and $\psi_s$ are completely positive.
Let us show that $\psi_n$ is idempotent. Let $K \in \hK(\hi)$ be a compact operator. Then
\begin{align}
    \psi_n(K)=\psi_0(K)=\psi_0(\Phi_0^*(K))\ge \psi_n(\Phi_0^*(K))
\end{align}
Thus $\psi_n-\psi_n \circ \Phi_0^*$ is positive when restricted to compact operators. Since $\psi_n-\psi_n \circ \Phi_0^*$ is normal, this implies that $\Vert \psi_n-\psi_n \circ \Phi_0^*\Vert= \Vert(\psi_n-\psi_n \circ \Phi_0^*)(I)\Vert =0$. Thus $\psi_n=\psi_n \circ \Phi_0^*$ and therefore $\psi_n=\psi_n \circ \left( \frac{1}{n}\sum_{k=0}^{n-1} (\Phi_0^*)^k \right)$. Since $\psi_n$ is normal, by Lemma \ref{lemmafixproj} we see that $\psi_n=\psi_n \circ \psi_0$. Thus for all compact operators $K \in \hK(\hi)$ we have $\psi_n(K)=\psi_n(\psi_n(K))$, which by normality implies $\psi_n=\psi_n \circ \psi_n$. \par We get the equality $\psi_n=\psi_0 \circ \psi_n$ similarly: for all compact operators $K$ we have the following.
\begin{align}
\psi_n(K)=\psi_0(K)=\Phi_0^*(\psi_0(K))=\Phi_0^*(\psi_n(K))
\end{align}
Thus by the essentially same argument as before, $\psi_n=\psi_0 \circ \psi_n$. Combining $\psi_n=\psi_0 \circ \psi_n$ and $\psi_n=\psi_n \circ \psi_0$ we see that 
\begin{align}
    \psi_s \circ \psi_s&=(\psi_0-\psi_n) \circ (\psi_0-\psi_n)\\
    &=\psi_0 \circ \psi_0-\psi_0\circ \psi_n-\psi_n \circ \psi_0 +\psi_n \circ \psi_n \\
    &=\psi_0-\psi_n=\psi_s
\end{align}
Thus $\psi_s$ is idempotent, and furthermore $\psi_n$ and $\psi_s$ are orthogonal: $\psi_n \circ \psi_s=\psi_s \circ \psi_n=0$. Thus 
\begin{align}
    \psi_0 =\psi_n \oplus \psi_s.
\end{align}
Finally, $\psi_n=\psi_0 \circ \psi_n$ implies that $\psi_n(\lh) \subset \fix{\Phi_0^*}$, and since $\psi_s=\psi_0-\psi_n$, we have that $\psi_s(\lh) \subset \fix{\Phi_0^*}$. Thus $\psi_n(\lh) \oplus \psi_s(\lh)=\fix{\Phi_0^*}$. \par 
Let us then show that $\psi_n(\lh)$ and $\psi_s(\lh)$ are orthogonal ideals of $\fix{\Phi_0^*}$ . Let $A \in \lh$ be an effect, $0\leq A \leq I$.
\begin{align}
    \psi_n(A)=\psi_n(A)\bullet I=\psi_n(A) \bullet \psi_n(I)+\psi_n(A) \bullet \psi_s(I)
\end{align}
Now $\psi_n(A) \bullet \psi_s(I)\leq \psi_n(A),\psi_s(I)$ so 
\begin{align}
\psi_n(A) \bullet \psi_s(I)&=\psi_0(\psi_n(A) \bullet \psi_s(I))\\
&=\psi_n(\psi_n(A) \bullet \psi_s(I))+\psi_s(\psi_n(A) \bullet \psi_s(I))\\
&\leq \psi_n(\psi_s(I))+\psi_s(\psi_n(A))=0
\end{align}
Thus $\psi_n(A) \bullet \psi_s(I)=0$ and so $\psi_n(A)=\psi_n(A)\bullet \psi_n(I)$. Similarly $\psi_s(A)=\psi_s(A) \bullet \psi_s(I)$. Since these hold for all effects, they also hold for all operators $A \in \lh$. Especially setting $A=I$ we see that $\psi_n(I)$ and $\psi_s(I)$ are projections in $\bullet$. Furthermore $A \in \psi_{n/s}(\lh) \Leftrightarrow A=A \bullet \psi_{n/s}(I)$. Thus finally we have the following.
\begin{align}
    \fix{\Phi_0^*}=\psi_n(I) \bullet \fix{\Phi_0^*} \oplus \psi_s(I) \bullet \fix{\Phi_0^*}
\end{align}
\end{proof}
Using this decomposition, we can give the following characterization of the broadcastable operators. We use the notation $\mathrm{supp}(\rho)$ for the support projection of a state $\rho \in \sh$.
\begin{theorem}\label{thmfixpchar}
    Let $\Phi_0^*$ be the marginal of a symmetric channel $\Phi^*$ and $\psi_0$ the map given by Lemma \ref{lemmafixproj}. Let $\psi_n,\psi_s$ be as in Proposition \ref{propsingornorm}. Then $A \in \fix{\Phi_0^*}$ if and only if there is a discrete subnormalized POVM $\{G_i\}_i \subset \psi_n(\lh)$ and states $\{\sigma_i\}_i$ s.t. $\sum_i G_i=\psi_n(I)$, $\tr{\sigma_i G_j}=\delta_{ij}$ and 
    \begin{align}\label{eqfixpointdecomp}
        A=\sum_i \tr{\sigma_iA}G_i +\psi_s(I)\bullet A
    \end{align}
\end{theorem}
\begin{proof}
    The "if"-direction is trivial. \par 
    According to Proposition \ref{propsingornorm} we have the following.
    \begin{align}
        A=\psi_n(I) \bullet A+\psi_s(I)\bullet A
    \end{align}
    Thus we focus on the normal part.
    \par Again by Proposition \ref{propsingornorm}, the normal part $\psi_n(I) \bullet \fix{\Phi_0^*}$ is a $W^*$-subalgebra of $(\fix{\Phi_0^*},\bullet)$ and the restriction of the product $(A,B) \mapsto \psi_0(\Phi^*(A \otimes B))$ to $\psi_n(I) \bullet \fix{\Phi_0^*}$ is induced by the normal map $\psi_n \circ \Phi^*:\hL(\hi \otimes \hi) \to \lh$, i.e.  $(A,B) \mapsto \psi_n(\Phi^*(A \otimes B))$. Furthermore, the $\psi_n(I)$ acts as a unit element for this algebra. Therefore by Proposition 3.6 in \cite{Kaniowski2014} we have that the von Neumann algebra $\psi_n(I) \bullet \fix{\Phi_0^*}$ is completely atomic. Thus there is a set of projections in $\bullet$, $\{G_i\} \subset \psi_n(I) \bullet \fix{\Phi_0^*}$ and normal states $\{\fii_i\}_i \subset (\psi_n(I) \bullet \fix{\Phi_0^*})_*$ such that for all $A \in \psi_n(I) \bullet \fix{\Phi_0^*}$ we have the following.
    \begin{align}
        A=\sum_i \fii_i(A) G_i
    \end{align}
    Especially $\fii_i(G_j)=\delta_{ij}$. Furthermore the linear functionals $\fii_i \circ \psi_n:\lh \to \mathbb{C}$ are ultraweakly continuous as a composition of two weakly$^*$-continuous functions. Also, $\fii_i(\psi_n(I))=1$ so all of them are states. Thus for all $i$ there exists a $\sigma_i \in \sh$ such that $(\fii_i \circ \psi_n)(A)=\tr{\sigma_i A}$ for all $A \in \lh$. Therefore, if $A \in \psi_n(I) \bullet \fix{\Phi_0^*}$, we have the following.
    \begin{align}
        A&=\sum_i \fii_i(A) G_i=\sum_i \fii_i(\psi_n(A)) G_i\\
        &=\sum_i \tr{\sigma_i A} G_i
    \end{align}
    Of course $\fii_i(G_j)=\delta_{ij}$ implies $\tr{\sigma_i G_j}=\delta_{ij}$. Thus we get the desired decomposition for all $A \in \fix{\Phi_0^*}$ as follows.
    \begin{align}
        A&=\sum_i \tr{\sigma_i (\psi_n(I) \bullet A)}G_i+\psi_s(I) \bullet A \\
        &=\sum_i \tr{\sigma_i  A}G_i+\psi_s(I) \bullet A 
    \end{align}
    Finally we show that the separability of the Hilbert space $\hi$ implies that there can be at most a countable amount of $G_i$ i.e. $\{G_i\}_i$ is discrete. This can be seen as follows. The condition $\tr{\sigma_i G_i}=1$ implies that $G_i\sigma_i=\sigma_i$ and furthermore $\tr{G_i \sigma_j}=0$ implies $G_i\sigma_j=0$ for $j \neq i$. Thus for all $i\neq j$ we have $\sigma_j\sigma_i=\sigma_j (G_i \sigma_i)=(G_i\sigma_j)^*\sigma_i=0$. Therefore, if  $\{G_i\}_i$ is uncountable, then there is an uncountable orthogonal set of states $\{\sigma_i\}_i$, which contradicts the separability of $\hi$.
\end{proof}
A discrete POVM $\{G_i\}_i$ with states $\{\sigma_i\}_i$ such that $\tr{G_i\sigma_j}=\delta_{ij}$ (i.e. $G_i\sigma_i=\sigma_i$), presented in the previous theorem, is occasionally called an \emph{eigenvalue-1 POVM}. In the case of an infinite-dimensional Hilbert space, this is not equivalent to a \emph{norm-1 POVM} i.e. $\Vert G_i\Vert=1, \; \forall i$. In the finite-dimensional case a norm-1 POVM is exactly an eigenvalue-1 POVM. \par 
Furthermore, if the Hilbert space $\hi$ is finite-dimensional, then there exist no singular linear functionals ($\lh$ is self-dual in this case) and so $\psi_0$ is automatically normal. Therefore in the finite-dimensional case we immediately see that all broadcastable operators $A$ are fixed points of the EB-channel $\psi_0(B)=\psi_n(B)=\sum_i \tr{\sigma_i B}G_i$
induced by the norm-1 POVM $\{G_i\}_i$. Thus one immediately derives the measurement-side results of \cite{jokinen2024nobroadcasting} from Theorem \ref{thmfixpchar}. 
\par 
To conclude this subsection, we show that there is an intimate connection between the Schrödinger picture fixed point states of $\Phi_0$ and the normal part $\psi_n=\sum_i \tr{\sigma_i \cdot}G_i$. Let us denote the Schrödinger dual of $\psi_n$ by $\psi_{n*}$
\begin{proposition}\label{propschroedfixnormal}
    Let $\Phi_0^*$ and $\psi_n$ be as defined above. Then for $\rho \in \sh$ we have the following.
    \begin{align}
        \Phi_0(\rho)=\rho \Leftrightarrow \psi_{n*}(\rho)=\rho
    \end{align}
    Especially, $\psi_n$ can be completed to a unital EB-channel preserving the fixed points:
    \begin{align}\label{eqschreb}
        \Phi_0(\rho)=\rho \Leftrightarrow \rho=\sum_i \tr{G_i\rho} \sigma_i +\tr{\psi_s(I)\rho}\sigma,
    \end{align}
    where $\sigma \in \sh$ is some state.
\end{proposition}
\begin{proof}
    Suppose $\psi_{n*}(\rho)=\rho$.
    Now $\Phi_0(\rho)=\rho$ since the proof of Proposition \ref{propsingornorm} implies $\psi_n \circ \Phi_0^*=\psi_n$, which in Schrödinger picture is just $\Phi_0 \circ \psi_{n*}=\psi_{n*}$. 
    \par Suppose then that $\Phi_0(\rho)=\rho$. Then for all compact operators $K$ we get the following.
    \begin{align}
        \tr{\rho K}=\tr{\rho \frac{1}{n}\sum_{k=0}^{n-1} (\Phi_0^*)^k(K)}
    \end{align}
    Thus by Lemma \ref{lemmafixproj} we have $\tr{\rho K}=\tr{\rho \psi_0(K)}=\tr{\rho \psi_n(K)}=\tr{\psi_{n*}(\rho)K}$. Since this holds for all compact operators, we have that $\psi_{n*}(\rho)=\rho$. We then get all $\rho \in \fix{\Phi_0}$ as fixed points of the EB-channel \eqref{eqschreb}, since $\tr{\rho \psi_s(I)}=0$. This follows from the fact that $\psi_s(I)=I-\psi_n(I)$ and $\tr{\rho \psi_n(I)}=1$.
\end{proof}

Proposition \ref{propschroedfixnormal} shows that the image of the Schrödinger picture of the normal part $\psi_n$ is a projection onto the trace-class fixed points of $\Phi_0$. Thus especially if $\Phi_0$ has no fixed point states, then $\psi_0$ is completely singular. \par 
 In infinite dimensional quantum systems, there indeed are channels with a completely singular fixed point projection $\psi_0$. An easy example of such a channel is the "shift-channel", defined as follows. Let $\{\ket{n}\}_{n \in \N}$ be an orthonormal basis of the infinite-dimensional separable Hilbert space $\hi$. The shift-channel $\hE_{\textrm{shift}}:\th \to \th$ is given by the following for all $T \in \th$.
\begin{align}\label{eqshiftchannel}
    \hE_{\textrm{shift}}(T):=\sum_{n \in \N} \ip{n}{T|n}\kb{n+1}{n+1}
\end{align}
It was shown in \cite{salzmann2024robustnessfixedpointsquantum} that this channel indeed has no trace-class fixed points other than 0. Therefore the normal part for this channel is just $\{0\}$. \par 
We give here another, non-discrete example of such a channel. 
In the following example, $\ket{z}=e^{-|z|^2/2}\sum_{n=0}^\infty \frac{z^n}{\sqrt{n!}}\ket{n}$ is the coherent state, defined in the separable Hilbert space $L^2(\R)$ with the orthonormal basis of Hermite functions $\{\ket{n}\}_{n=0}^\infty$.

\begin{example}\label{exqfunction}
    Let $\hi=L^2(\R)$ and $\Lambda^*:\lh \to \lh$ a map defined as $\Lambda^*(T)=\frac{1}{\pi}\int_{\mathbb{C}} \kb{z}{z} \, \ip{z}{T | z} \, dz$ for all $T \in \lh$  (i.e. the "$Q$-function" EB-channel). Then $\fix{\Lambda^*}=\mathbb{C}I$.
\end{example}

\begin{proof}
   The fixed point condition now reads out as 
   \begin{align}\label{eqqfunctfix}
       T=\frac{1}{\pi}\int_{\mathbb{C}} \kb{z}{z} \, \ip{z}{T | z} \, dz.
   \end{align}
   From this it follows that for all $w \in \mathbb{C}$
   \begin{align}\label{eqconvolution1}
       \ip{w}{T|w}=\frac{1}{\pi}\int_{\mathbb{C}} |\ip{z}{w}|^2 \, \ip{z}{T | z} \, dz 
   \end{align}
   Let us define $f(w):=\ip{w}{T|w}$. Thus equation \eqref{eqconvolution1} implies that 
   \begin{align}\label{eqconvolution2}
       f(w)=\int_{\mathbb{C}} \frac{e^{-|w-z|^2}}{\pi} \, f(z) \, dz
   \end{align}
   Thus $f=f *g$, where $g(z):=\frac{e^{-|z|^2}}{\pi} $ and $*$ denotes convolution. Therefore also $f=f*g^n$ for all $n \in \N$, where $g^n$ is the $n$-fold convolution of $g$ with itself. The idea is to show that this repeated convolving with a Gaussian "flattens" $f$ into a constant function. Then, since $f$ is constant, we have that for all $w \in \mathbb{C}$, $$c=f(w)=\ip{w}{T|w}$$ Thus equation \eqref{eqqfunctfix} implies that $T=cI$, finishing the proof. Showing that $f$ is constant is routine calculus, and thus the calculation is relegated to the Appendix. 
\end{proof}
The channel in the previous Example is self-adjoint in the sense that $\Lambda^*=\Lambda$, when restricted to the trace-class. Furthermore the only multiple of identity that is trace-class in an infinite-dimensional space is 0, thus this channel has no fixed point states. Hence, even though the fixed point projection is singular, the broadcasting algebra is still completely atomic.

We call $\Lambda$ the "$Q$-function" EB-channel because it is in fact related to the Husimi $Q$-representation \cite{husimi1940some} defined as $P_\rho (z) = \frac{\bra{z} \rho \ket{z}}{\pi}$. The argument is as follows: for any operator of the form $\int f(z) \dyad{z} dz$ and $w \in \mathbb{C}$, with $f$ bounded and continuous, we can define the map $A_w$ such that $A_w(\int f(z) \dyad{z} dz) = f(w)$. Formally $A_w$ is defined by $A_w(\dyad{z}) = \delta(z-w)$ where the right hand side is the Dirac delta function. We will show how the map $A_w$ can be defined rigorously below. Then $A_w (\Lambda(\rho)) = \frac{\bra{w} \rho \ket{w}}{\pi}$ is nothing but the Husimi $Q$-representation of $\rho$ evaluated at $w \in \mathbb{C}$. As also noted in \cite{linowski2024relating}, it is not straightforward to map the Husimi $Q$-representation of a state to back to its density matrix, or to its Wigner function. Hence one should not expect the channel $\Lambda$ to be invertible, let alone to have nontrivial fixed points.

Let us proceed to show that the map $A_w$ is well defined. The POVM $$G(X)=\frac{1}{\pi} \int_X \kb{z}{z} \, dz$$ is known to be extremal \cite{Heinosaari_2012}. Hence $$\int_\mathbb{C} f(z) \, dG(z)=\int_\mathbb{C} g(z) \, dG(z)$$ implies $f=g$ \cite{Pello_extreme}. This can be alternatively also seen with operator convolutions by using Proposition 1 and Proposition 4 in \cite{kiukas2012}. Thus the function $A_w$ is well-defined for all $w \in \mathbb{C}$. 

\subsubsection{Proofs of Subsections "Contextuality and broadcasting for quantum states" and "Contextuality and broadcasting for quantum measurements"}
We finally present the proofs of Sections "Contextuality and broadcasting for quantum states" and "Contextuality and broadcasting for quantum measurements".
\\
\textbf{Proof of Theorem \ref{thmstatecont}.}

    $\eqref{c_ont} \Leftrightarrow \eqref{c_eb}$ follows directly from Definition \ref{defnormal} as previously shown. 
    
    $\eqref{c_eb} \Rightarrow \eqref{c_br}$ follows from the fact that EB-channels can be broadcast with the channel \eqref{eqsymmetriceb}. 
    
    $\eqref{c_eb} \Leftarrow \eqref{c_br}$ follows from Proposition \ref{propschroedfixnormal}. 

    Finally the equivalence $\eqref{c_br} \Leftrightarrow \eqref{c_com}$ is Theorem 3.14 in \cite{Kaniowski2014}. 
    \qed
\\
\textbf{Proof of Theorem \ref{thmmeaschar}.}

$\eqref{d_ont} \Leftrightarrow \eqref{d_eb} \Rightarrow \eqref{d_br}$ follow from Definition \ref{defnormal} and broadcasting with the channel in equation \eqref{eqsymmetriceb}. \par 
    $\eqref{d_br} \Leftrightarrow \eqref{d_dec}$ follows from Theorem \ref{thmfixpchar} by choosing $M_n(X)=\psi_n(I)\bullet M(X)$ and $M_s(X)=\psi_s(I) \bullet M(X)$ for all $X \in \hB(\Omega_M)$. We should note however, that $M_s$ is indeed ultraweakly $\sigma$-additive here, since the product of a von Neumann algebra is always separately weakly$^*$-continuous \cite[Theorem 1.7.8]{Sakai1998} (and the predual of $\fix{\Phi_0^*}$ is just the trace-class quotiented with the annihilator $\fix{\Phi_0^*}^\perp$ \cite[Proposition V.2.12]{198897}). \par 
    Furthermore, if $M_s=0$ for some $M \in \hM$, then we have $M=M_n$. This especially implies that $I=M(\Omega_M)$ is in the normal part of the broadcasting algebra. Thus if $\psi_n$ is the projection to the normal part given in Proposition \ref{propsingornorm}, we have that $\psi_n(I)=I$ i.e. $\psi_s=0$. Thus equation \eqref{eqfixpointdecomp} gives an EB-channel, showing $\eqref{d_dec} \Rightarrow \eqref{d_eb}$. This makes all the statements equivalent. \qed
\\
\textbf{Proof of Corollary \ref{corpvmdiscrete}.}

    Let $\hN$ be the (concrete) von Neumann algebra generated by the PVMs in $\hM$ and let $P,Q \in \hN$ be projections. Let us show that $QP=Q \bullet P$ so that $QP \in \fix{\Phi_0^*}$ and thus $\hN \subset \fix{\Phi_0^*}$. 
    \begin{align}
        QP=(Q\bullet P)P+(Q \bullet P^\perp)P
    \end{align}
    Now $(Q \bullet P^\perp)P=0$, which is seen as follows. As a reminder, the broadcasting product $\bullet$ is bipositive, so $0\leq Q \bullet P^\perp \leq I \bullet P^\perp=P^\perp$. Therefore we have that $$P(Q \bullet P^\perp)^2P\leq PP^\perp P=0$$ and thus 
    \begin{align}
    |(Q \bullet P^\perp)P|^2&=((Q \bullet P^\perp)P)^*((Q \bullet P^\perp)P)\\
    &=P(Q \bullet P^\perp)^2P=0. 
    \end{align}
    Hence taking the polar decomposition we see that $$(Q \bullet P^\perp)P=V|(Q \bullet P^\perp)P|=0.$$ We also similarly see that $ (Q \bullet P)P^\perp=0$. Thus
    \begin{align}
        PQ=P(Q\bullet P)=Q\bullet P-P^\perp(Q\bullet P)=Q\bullet P
    \end{align}
    Thus $\hN$ is a subalgebra contained within the fixed points and so by Theorem \ref{thmlhalgebra} is completely atomic. Therefore there exists a discrete generating PVM $\{P_i\}_i$ i.e. for all projections $P \in \hN$ we have 
    \begin{align}\label{eqpureatom}
        P=\sum_{i \in I_P} P_i
    \end{align}
    Note that since $\hi$ is separable, the index set can at most be countably infinite. 
    Now picking for every $i$ a state $\sigma_i$ supported by $P_i$ we get that all projections $P \in \hN$ can be given as fixed points of the EB-channel $A \mapsto \sum_i \tr{\sigma_i A}P_i$. This shows \eqref{d_br} $\Rightarrow$ \eqref{d_eb} in Theorem \ref{thmmeaschar}.\par 
    Furthermore, if $X \mapsto P(X)$ is a PVM in $\hM$, then the von Neumann algebra generated by $X \mapsto P(X)$ is a subalgebra of $\hN$, so by Theorem \ref{thmlhalgebra} and separability of $\hi$,  $P$ is a purely atomic PVM. Of course one can also see this easily from equation \eqref{eqpureatom}. Since the measurable space $(\Omega_P,\hB(\Omega_P))$ is assumed to be standard Borel, this means that $\Omega_P$ is countable and thus $P$ is a discrete operator measure (see e.g. Theorem 2 of Section 2.1.6 in \cite{Kadets2018}).

\qed

\section{V. Connection between the approximate and the non-normal definition}

As was shown in the previous sections, defining noncontextual ontological models consistent with usual probability theory leads to an unavoidable discreteness condition for sharp measurements. Thus if one wishes that commuting measurements, such as position, are contextuality non-confirming, some assumptions must be relaxed. One option for this is to give up the normality of the measurement response functions in Definition \ref{defnonnormal}.

Let us examine Definition \ref{defnonnormal} more closely. This is again a measure-and-prepare scenario, where the measurement responses are now given a degree of freedom over null sets i.e. all properties of $T$ are defined almost everywhere. In other words, when $\rho$ is measured with the measurement $G$, all observable outcomes must be in a set with positive measure $X$ and the state $\omega_X^\rho(A)=\frac{1}{\tr{\rho G(X)}} \int_X T(A)(\lambda) \, d\tr{\rho G(\lambda)}$ is prepared. Since there is no additional continuity assumptions for $T$, we conclude that the prepared state $\omega_X^\rho$ can be a non-normal state in $\hS(\lh)$. One could also require the state preparation to be pointwise i.e. $A \mapsto T(A)(\lambda)$ is a state for all $\lambda$, (which is a special case of Definition \ref{defmeasresp}), but we focus on the more natural definition, where the properties of $T$ are not in general defined pointwise. 

In the familiar discrete case i.e. $\Omega=\{\lambda_1, \dots, \lambda_n\}$ there are no (nontrivial) null sets, so the measurement responses of Definition \ref{defmeasresp} are exactly the usual measurement responses: for a POVM $\{M_k\}_k$ we may freely set $p(k|M,\lambda_n)=T(M_k)(\lambda_n)$ to get an equivalent model.  Furthermore, assuming that $A \mapsto T(A)$ is normal (i.e. continuous from the weak$^*$-topology of $\lh$ to the weak$^*$-topology of $L^\infty$) we recover the pointwise state preparation present in Definition \ref{defnormal} of contextuality non-confirming sets. This is what we prove next. This then especially implies that Definition \ref{defnonnormal} is consistent with Definition \ref{defnormal}. 
\par We begin the proof by first proving the following useful lemma.
\begin{lemma}\label{lemmastateextension}
    Let $(\Omega,\hB(\Omega),\mu)$ be a standard Borel measure space, $T:\lh \to L^\infty(\Omega,\mu)$ a positive linear map with $T(I)=1$ and $\hD \subset \lh$ a norm-separable, norm-closed and self-adjoint (i.e. $A \in \hD \Leftrightarrow A^* \in \hD$) subspace. Then for every countable dense set $C$ of $\hD$ there exists a $N_C \subset \Omega$ with $\mu(N_C)=0$ and a family of states $\{\omega_\lambda\} \subset \hS(\lh)$ such that $T(A)(\lambda)=\omega_\lambda(A)$ for all $A \in C$ and $\lambda \notin N_C$.
\end{lemma}

\begin{proof}
    Consider the subspace $S=\mathrm{lin}(\hD \cup \{I\})$. Then $S$ is also separable, since if $\{A_n\}_{n \in \N}$ is a countable dense set of $\hD$, then $(\{A_n\}_{n \in \N} + (\Q+i\Q)I)$ is a countable dense set in $S$. The goal is to use Arveson's extension theorem for operator systems to find the states $\omega_\lambda$ and adjoining the identity this way makes $S$ into an operator system. \par 
    Let thus $C_1$ be some countable dense set of $S$.
    Furthermore, since $S$ is separable, so is $S \cap \lh_+$ as a metric subspace of $S$. Let thus  $C_2$ be some countable dense set in $S \cap \lh_+$. We can then define the countable dense set $C=C_1 \cup C_2$. \par 
    We first consider continuity. Now $T$ is bounded since it is positive. Thus there exists a $K>0$ such that for each $A\in \lh$ we have $|T(A)(\lambda)|\leq K\Vert A\Vert$ for almost every $\lambda$. Especially if $C=\{A_n\}_{n \in \N}$ is the above specified countable dense set, then the set of complex-rational linear combinations of $C$ is a countable set, denoted by $\mathrm{lin}_{\Q+i\Q}(C)$. Then we have that
    \begin{align}
        |T(D_n)(\lambda)|\leq K \Vert D_n \Vert, \quad \lambda \notin M_n, \quad \mu(M_n)=0
    \end{align}
    for all $D_n \in \mathrm{lin}_{\Q+i\Q}(C)$.
    Thus especially $|T(A_n-A_m)(\lambda)|\leq K\Vert A_n-A_m\Vert$ for all $n,m$, when $\lambda \notin N_1:=\bigcup_{n \in \N} M_n$, where $N_1$ is a null set by the subadditivity of the measure $\mu$. Thus for all $\lambda \notin N_1$, $A \mapsto T(A)(\lambda)$ is uniformly continuous on $C$. It is well known that for uniformly continuous functions, from a dense set to a Banach space, there is a unique continuous extension to the whole space. Let thus $\phi_\lambda:S \to \mathbb{C}$ be this extension for all $\lambda \notin N_1$.  \par 
    Let us then consider positivity. For this, let  $C_2=\{B_n\}_{n \in \N}$. Since $C_2 \subset C$, we have for all $n$ and $\lambda \notin N_1$ that $\phi_\lambda(B_n)=T(B_n)(\lambda)$. The fact that $T$ is a positive linear map implies that for all $n \in \N$ we have $T(B_n)(\lambda)\ge 0$ for $\lambda \notin L_n$, $\mu(L_n)=0$. Thus $T(B_n)(\lambda)\ge 0$ holds for all $n \in \N$ when $\lambda \notin N_2:=\bigcup_{n \in \N} L_n$. Then $N_2$ is a null set and furthermore $N_1 \cup N_2$ is a null set, both by subadditivity. Let then $\lambda \notin N_1 \cup N_2$ and $A \in S \cap \lh_+$. Then there is a subsequence $\{B_{n_k}\}_{k \in \N} \subset C_2$ such that $A=\lim_{k \to \infty } B_{n_k}$. Therefore, since $\phi_\lambda$ is continuous, we have the following.
    \begin{align}
       \phi_\lambda(A)=\lim_{k \to \infty} \phi_\lambda(B_{n_k})=\lim_{k \to \infty} T(B_{n_k})(\lambda)\ge 0
    \end{align}
    Thus for all $\lambda \notin N_1 \cup N_2$ the map $\phi_\lambda$ is positive in $S$. \par 
    Let us then consider linearity. Consider the countable set $C^2 \times (\Q+i\Q)^2$. Now by $L^\infty$-linearity, the following holds for all $(A_n,A_m,q_k,q_l) \in C^2 \times (\Q+i\Q)^2$.
    \begin{align}\label{eqlineardense}
        T(q_kA_n+q_lA_m)&(\lambda)=q_kT(A_n)+q_lT(A_m) \\
        &\mathrm{for}\;\lambda \notin N_{(A_n,A_m,q_k,q_l)}
    \end{align}
    Thus we can define the null set $N_3:=\bigcup_{\mathbf{a} \in C^2 \times (\Q+i\Q)^2} N_{\mathbf{a}}$ so that the equation \eqref{eqlineardense} holds for all $(A_n,A_m,q_k,q_l) \in C^2 \times (\Q+i\Q)^2$
    when $\lambda \notin N_3$. Suppose then that $\lambda \notin N_1 \cup N_2 \cup N_3$, $A,B \in S$ and $a_1,a_2 \in \mathbb{C}$. Then there exist subsequences from $C$ converging to $A,B$ and subsequences from $(\Q+i\Q)$ converging to $a_1$ and $a_2$. Thus since when $\lambda \notin N_1$, $\phi_\lambda$ is continuous, we can use continuity four times to deduce that $\phi_\lambda(a_1A+a_2B)=a_1\phi_\lambda(A)+a_2\phi_\lambda(B)$ when $\lambda \notin N_1 \cup N_3$. Thus if $\lambda \notin N_1 \cup N_2 \cup N_3$, we have that $\phi_\lambda$ is bounded, positive and linear.
     \par 
    Finally, we can pick a null set $N_4$ such that $T(I)(\lambda)=1$ when $\lambda \notin N_4$. Thus by defining $N= N_1 \cup N_2 \cup N_3 \cup N_4$ we see that $\phi_\lambda$ is a positive normalized linear functional on the operator system $S$ for all $\lambda \notin N$. Then especially $\phi_\lambda$ is a completely positive map from the operator system $S$ to $\mathbb{C}$ for all $\lambda \notin N$. Thus we can use Arveson's extension theorem \cite[Theorem 7.5]{Paulsen_2003} to extend $\phi_\lambda$ to a state $\omega_\lambda \in \hS(\lh)$ for all $\lambda \notin N$. Furthermore for $A \in C$ we have $\omega_\lambda(A)=\phi_\lambda(A)=T(A)(\lambda), \; \lambda \notin N$, finishing the proof. 
\end{proof}

\begin{proposition}\label{propdefnormalreduction}
Suppose that in Definition \ref{defnonnormal} $T:\lh \to L^\infty(\Omega,\tr{\rho_0G})$ is normal. Then there exists a weakly measurable set of states $\{\sigma_\lambda\}_{\lambda \in \Omega}$ such that the following holds for all $A \in \lh$ and $\rho \in \sh$.
\begin{align}
    \int_\Omega T(A)(\lambda) d\tr{\rho G(\lambda)}=\int_\Omega \tr{\sigma_\lambda A} d\tr{\rho G(\lambda)}
\end{align}
Especially Definition \ref{defnormal} is a special case of Definition \ref{defnonnormal}.
\end{proposition}
\begin{proof}
    If $\dim \hi <\infty$, the result is trivial, so assume that $\dim \hi =\infty$.
    It is well known, that the set of compact operators $\hK(\hi)$ is norm-separable when $\hi$ is a separable Hilbert space. Let $\{\ket{n}\}_{n\in \N}$ be an orthonormal basis of $\hi$ and let $C$ be some countable dense set of the compact operators containing $\{\kb{m}{n}\}_{m,n \in \N}$. Thus by Lemma \ref{lemmastateextension} there is a null set $N$ and states $\{\omega_{\lambda}\}_{\lambda \notin N}$ such that $T(\kb{m}{n})(\lambda)=\omega_\lambda(\kb{m}{n})$ for all $n,m$. Let thus $\lambda \notin N$. Now by Theorem III.2.14 in \cite{Takesaki1979} $\omega_{\lambda}$ can be divided to normal and singular parts $\omega_{\lambda}=\omega_{\lambda,norm}+\omega_{\lambda,sing}$.
    Let $T_{\lambda} \in \th$ be such that $\omega_{\lambda,norm}=\tr{T_\lambda (\cdot)}$. Now since singular linear functionals map all compact operators to zero, we have the following.
    \begin{align}
        T(\kb{m}{n})(\lambda)=\omega_\lambda(\kb{m}{n})=\tr{T_\lambda \kb{m}{n}}
    \end{align}
    Thus especially setting $T_\lambda$ to some state in the null set $N$, we get a measurable function $\lambda \mapsto \tr{T_\lambda \kb{m}{n}}$ for each $m,n \in \N$. Furthermore, $T_\lambda \ge 0$, and since $\tr{T_\lambda (\cdot )}\leq \omega_\lambda$, we have $\tr{T_\lambda}\leq 1$. \par 
    Let then $A \in \lh$ be arbitrary. Then ultraweakly
    $$A=\lim_{n \to \infty} P_nAP_n=\lim_{n \to \infty} \sum_{j,k}^n \ip{j}{A|k}\kb{j}{k}
    $$ where $P_n=\sum_{k=1}^n |k\rangle\langle k|$. Thus $\tr{T_\lambda A}=\lim_{n \to \infty} \sum_{j,k}^n \ip{j}{A|k}\tr{T_\lambda \kb{j}{k}}$, so that $\lambda \mapsto \tr{T_\lambda A}$ is measurable as a pointwise limit of measurable functions for all $A \in \lh$. Furthermore for all $\lambda \notin N$ we have $|\tr{T_\lambda P_nAP_n}|\leq \Vert A\Vert$, so that the measurable functions $\lambda\mapsto\tr{T_\lambda P_nAP_n}$ are  dominated by an integrable (constant) function. Let then $\rho \in \sh$ be any state. Since $\tr{\rho G}$ is absolutely continuous w.r.t. $\tr{\rho_0G}$ (as $\rho_0$ is faithful), by the Radon-Nikodym theorem there is a $f_\rho \in L^1(\Omega,\tr{\rho_0 G})$ such that $\tr{\rho G(X)}=\int_X f_\rho(\lambda) \, d\tr{\rho_0 G(\lambda)}$. Thus especially $f \mapsto \int_\Omega f(\lambda) \, d(\tr{\rho G(\lambda)})$ is a $L^{\infty}(\tr{\rho_0 G})$-weak$^*$-continuous linear functional. Since $T$ is normal, $T(A)=\lim_{n\rightarrow\infty} T(P_nAP_n)$ in $L^\infty(\Omega,\tr{\rho_0 G})$-weak$^*$ sense, and hence the following holds for all $\rho \in \sh$:
    \begin{align}
        &\int_\Omega T(A)(\lambda) \, d\tr{\rho G(\lambda)}\\
        &=\lim_{n \to \infty } \int 
        _\Omega T(P_nAP_n)(\lambda) \, d\tr{\rho G(\lambda)} \\
        &=\lim_{n \to \infty } \int 
        _\Omega \tr{T_\lambda P_nAP_n} d\tr{\rho G(\lambda)}
        \\
        &\overset{DCT}{=} \int 
        _\Omega \tr{T_\lambda A} \,d\tr{\rho G(\lambda)}
    \end{align}
    Especially setting $A=I$ and $\rho=\rho_0$, we see that $\tr{T_\lambda}=1$ almost everywhere. Thus the $\{T_\lambda\}_\lambda$ can be modified within a set of measure zero to be a weakly measurable set of states $\{\sigma_\lambda\}_\lambda$, finishing the proof.
\end{proof}
Thus the different definitions are indeed consistent with each other. We now show that, when considering states, also the contextuality-confirming sets of states are the same for Definitions \ref{defnormal} and \ref{defnonnormal}. This further indicates that the normal definition is enough when considering only states.

\begin{proposition}\label{propdefequivalence}
    Let $\hS \subset \sh$ be a set of states. Then $\hS$ is contextuality non-confirming for all measurements in the sense of Definition \ref{defnormal} if and only if it is contextuality non-confirming for all measurements in the sense of Defintion \ref{defnonnormal}. 
\end{proposition}
\begin{proof}
    Definition \ref{defnonnormal} is an extension of Definition \ref{defnormal} so one of the directions is trivial. We focus on the other direction.
    \par
    Let $\{\ket{n}\}_{n=1}^{\dim\hi}$ be an orthonormal basis of $\hi$.
    One can repeat the argument of Proposition \ref{propdefnormalreduction} (until normality was used) to find a weakly measurable family of $T_\lambda \in \th$ such that for all $\lambda \in \Omega$ we have $T(\kb{m}{n})(\lambda)=\tr{T_\lambda \kb{m}{n}}$ for all $m,n \in \N$, $T_\lambda \ge 0$ and $\tr{T_\lambda}\leq 1$. Let then $\rho \in \hS$. Then we have the following for all $m,n \in \N$
    \begin{align}
        \tr{\rho \kb{m}{n}}&=\int_\Omega T(\kb{m}{n})(\lambda) \, d\tr{\rho G(\lambda)}\\
        &=\int_\Omega \tr{T_\lambda\kb{m}{n}} \, d\tr{\rho G(\lambda)}\label{eqstatesandwich}
    \end{align}
    Then especially we get the following using the monotone convergence theorem. 
    \begin{align}
        1&=\tr{\rho}=\sum_{n \in \N}\tr{\rho \kb{n}{n}} \\
        &=\sum_{n \in \N}\int_\Omega \tr{T_\lambda\kb{n}{n}} \, d\tr{\rho G(\lambda)} \\
        &=\int_\Omega \tr{T_\lambda} \, d\tr{\rho G(\lambda)}
    \end{align}
    Thus, since $\tr{T_\lambda}\leq 1$, we have that $\tr{T_\lambda}=1$ for $\tr{\rho G}$-almost every $\lambda$. 
    \par 
Now, since the trace-class is a separable metric space with respect to the trace-norm, $\hS$ is separable as well, as a metric subspace of the trace-class. Therefore there exists a countable dense set $\{\rho_n\}_n \subset \hS$ in the trace-norm topology of $\hS$. Let us define the state $\Tilde{\rho}=\sum_n t_n \rho_n$ with $t_n>0$ and $\sum_n t_n=1$. Now for all $k<\#\{\rho_n\}_n$ and $A \in \lh$ we have that 
\begin{align}\label{eqlimits}
    &\tr{\left(\sum_{n=1}^k t_n \rho_n \right)A}\\&=\sum_{n=1}^k t_n\tr{\rho_nA} \\
    &= \sum_{n=1}^k t_n \int_\Omega T(A)(\lambda) \, d\tr{\rho_n G(\lambda)} \\
    &=\int_\Omega T(A)(\lambda) \, d\tr{\left(\sum_{n=1}^k t_n \rho_n \right) G(\lambda)}\label{eqlimitsf}
\end{align}
Thus if $\#\{\rho_n\}_n <\infty$, we see that $\tilde{\rho}$ is a fixed point.
Otherwise  $\left|\int_\Omega T(A)(\lambda) \, d\tr{\left(\tilde{\rho}-\sum_{n=1}^k t_n \rho_n \right) G(\lambda)}\right|\leq \Vert A \Vert \tr{\left(\sum_{n=k+1}^N t_n \rho_{n} \right) G(\Omega)}=\Vert A \Vert\sum_{n=k+1}^N t_n$. Thus if $\#\{\rho_n\}_n=\infty$, we see that $\sum_{n=k+1}^N t_n \to 0$ as $k \to \infty$. In this case taking the limits in equation \eqref{eqlimits}-\eqref{eqlimitsf}  (using that $\sum_n t_n \rho_n$ is trace norm convergent) we see the following:
\begin{align}
    \tr{\tilde{\rho} A}=\int_\Omega T(A)(\lambda) \, d\tr{\tilde{\rho} G(\lambda)}
\end{align}
Therefore by what we showed before, $\tr{T_\lambda}=1$ for all $\lambda \notin N_1$, for some $N_1$ such that $\tr{\tilde{\rho}G(N_1)}=0$. We now show that this implies that $\tr{\rho G(N_1)}=0$ for all $\rho \in \hS$. Since $t_n \tr{\rho_n G(N_1)}\leq \tr{\tilde{\rho}G(N_1)}=0$, we have that $\tr{\rho_nG(N_1)}=0$ for all $\rho_n$ in the dense set. Let then $\rho$ be arbitrary. Then there is a trace-norm convergent subsequence $\{\rho_{n_k}\}_k$ of $\{\rho_n\}_n$ with $\lim_{k \to \infty}\rho_{n_k}=\rho $. Thus also $\tr{\rho G(N_1)}=\lim_{k \to \infty}\tr{\rho_{n_k}G(N_1)}=0$. We have now shown that there is a common null set $N_1$ for all $\rho \in \hS$ such that $\tr{T_\lambda}=1$ when $\lambda \notin N_1$. Therefore $T_\lambda$ are states when $\lambda \notin N_1$. Furthermore, when $\lambda \in N_1$, we may set $T_\lambda=\sigma$, where $\sigma$ is some state. Thus every $\rho \in \hS$ is a fixed point of an EB-channel by equation \eqref{eqstatesandwich}. Thus we get the form given in Definition \ref{defnormal} for noncontextual ontological model, finishing the proof.
\end{proof}
We conclude that for sets of states it makes no difference which definition is used.  

\subsection{Approximate interpretation}
Definition \ref{defnonnormal} does not allow for an easy physical interpretation of the model as e.g. singular states and non-pointwise quantities are allowed. We now intend to show that these types of models can still be, in a sense, given as limits of approximate normal noncontextual models as given by Definition \ref{defapprox}.

Definition \ref{defapprox} means that every finite set of effects is approximately contextuality non-confirming with respect to some discrete ontological within the range of a POVM to arbitrary precision. We now show that this is equivalent to Definition \ref{defnonnormal}. We begin by proving the following lemma. This is essentially a $L^\infty$-valued and non-normal version of the extension in the main result of \cite{busch03}, but we prove it here for the reader's convenience. Let $\hE(\hi)$ be the set of the effects of $\lh$.
\begin{lemma}\label{lemmagleasonlike}
    Let $(\Omega,\hB(\Omega),\mu)$ be a standard Borel space and $T:\hE(\hi) \to L^\infty(\Omega,\mu)$ such that 
    \begin{enumerate}
        \item $T(I)=1$
        \item $T(E+F)=T(E)+T(F)$ if $E+F\leq I$
        \item $T(tE)=tT(E)$ for $0\leq t \leq 1$, $E \in \hE(\hi)$
        \item $0\leq T(E) \leq 1$.  
    \end{enumerate}
    Then $T$ extends to a positive linear map $\Tilde{T}:\lh \to L^\infty(\Omega,\mu)$ with $\tilde{T}(I)=1$
\end{lemma}
\begin{proof}
    Let first $A \ge 0$ be a positive operator. Then define 
    \begin{align}
        T'(A):=\Vert A \Vert T\left( \frac{A}{\Vert A \Vert}\right)
    \end{align}
    Note first that $T'(A)$ is measurable, since it is just a scalar multiple of a $L^\infty$-function. Furthermore $\Vert T(A) \Vert_{L^\infty}\leq \Vert A \Vert$, so $T'(A) \in L^\infty(\Omega,\mu)$. Let then $A,B\ge 0$
    \begin{align}
        T'(A+B)&=\Vert A+B\Vert T\left( \frac{A+B}{\Vert A+B \Vert}\right)\\
        &=\Vert A+B\Vert T\left( \frac{A}{\Vert A+B \Vert}\right) \\
        &+\Vert A+B\Vert T\left( \frac{B}{\Vert A+B \Vert}\right) \\
        &=\Vert A+B\Vert T\left(\frac{\Vert A \Vert}{\Vert A+B \Vert} \frac{A}{\Vert A \Vert}\right)\\
        &+\Vert A+B\Vert T\left(\frac{\Vert B \Vert}{\Vert A+B \Vert} \frac{B}{\Vert B \Vert}\right) \\
        &=\Vert A\Vert T\left(\frac{A}{\Vert A \Vert}\right)+\Vert B\Vert T\left(\frac{B}{\Vert B \Vert}\right)\\
        &=T'(A)+T'(B)
    \end{align}
    Furthermore, if $\lambda\ge 0$ we have the following.
    \begin{align}
        T'(\lambda A)=\lambda\Vert A\Vert T\left( \frac{\lambda A}{\lambda \Vert A \Vert }\right)=\lambda T'(A)
    \end{align}
    Thus $T'$ is conic-linear for the positive operators. \par 

    Let then $A$ be a self-adjoint (Hermitian) operator and $A=A^+-A^-$ the canonical decomposition to positive and negative parts (i.e. the one from the spectral decomposition of $A$). Let $A=B-C$ be another decomposition to positive and negative parts. We show that $T'(A^+)-T'(A^-)=T'(B)-T'(C)$. This follows from the fact that $A^+-A^-=B-C$ so that $A^++C=B+A^-$ so that 
    \begin{align}
        T'(A^+)+T'(C)&=T'(A^++C)=T'(B+A^-)\\
        &=T'(B)+T'(A^-),
    \end{align}
    which gives us $T'(A^+)-T'(A^-)=T'(B)-T'(C)$. Thus we may define $T''(A):=T'(A^+)-T'(A^-)$, where the value of $T''(A)$ does not depend on the decomposition of $A$ into positive and negative parts. Now $T''(A)$ is in $L^\infty(\Omega,\mu)$, since it is just a difference of two $L^\infty$-functions. \par 
    Let then $A,B$ both be self-adjoint. Then a decompostion of $A+B$ into positive and negative parts is $A+B=(A^++B^+)-(A^--B^-)$. Since the value of $T''(A+B)$ is independent of the decomposition, we have the following.
    \begin{align}
        &T''(A+B)\\
        &=T'(A^++B^+)-T'(A^-+B^-)\\
        &=T'(A^+)-T'(A^-)+T'(B^+)+T'(B^-) \\
        &= T''(A)+T''(B)
    \end{align}
    Let then $a \in \R$. Then if $a\ge 0$
    \begin{align}
        T''(aA)=T'(aA^+)-T'(aA^-)=aT''(A)
    \end{align}
    Furthermore, if $a<0$ we have the following.
    \begin{align}
        T''(aA)&=T'(|a|A^-)-T'(|a|A^+)\\
        &=|a|(T'(A^-)-T'(A^+))\\
        &=a(T'(A^+)-T'(A^-))\\
        &=aT''(A)
    \end{align}
Thus $T'':\lhs \to L^\infty(\Omega,\mu)$ is a (real) linear mapping. Furthermore it's bounded since $\Vert T''(A)\Vert_{L^\infty}=\Vert T'(A^+)-T'(A^-) \Vert_{L^\infty}\leq \Vert T'(A^+)\Vert_{L^\infty}+\Vert T'(A^-)\Vert_{L^\infty} \leq \Vert A^+ \Vert+\Vert A^-\Vert\leq 2\Vert A \Vert$. \par 
Finally let $A \in \lh$ be arbitrary. Then we have the decomposition $A=\frac{1}{2}(A+A^*)+i\left(\frac{-i}{2}(A-A^*)\right)$ into a linear combination of self-adjoint operators. Let us thus define $\tilde{T}(A):=\frac{1}{2}(T''(A+A^*)+iT''(-i(A-A^*))$. This is again in $L^\infty$ as a linear combination of $L^\infty$-functions. Then for $A,B \in \lh$ we have the following.
\begin{align}
    &\tilde{T}(A+B)\\
    &=\frac{1}{2}T''(A+B+A^*+B^*)\\
    &+\frac{i}{2}T''(-i(A+B-A^*-B^*) \\
    &=\frac{1}{2}(T''(A+A^*)+T''(B+B^*))\\
    &+\frac{i}{2}(T''(-i(A-A^*))+T''(-i(B-B^*)))\\
    &=\tilde{T}(A)+\tilde{T}(B)
\end{align}
Let the $z \in \mathbb{C}$ be arbitrary with $z=\Re(z)+i\Im(z)$. Then by the real-linearity of $T''$ we have the following. 
\begin{align}
    &\tilde{T}(zA)\\
    &=\frac{1}{2}(T''(zA+\overline{z}A^*)+iT''(-i(zA-\overline{z}A^*)) \\
    &=\frac{1}{2}(T''(\Re(z)(A+A^*)-\Im(z)(-i(A-A^*)))\\
    &+\frac{i}{2}T''(-i\Re(z)(A-A^*)+\Im(z)(A+A^*))) \\
    &=\frac{1}{2}((\Re(z)+i\Im(z))T''(A+A^*)\\
    &+\frac{1}{2}(i\Re(z)-\Im(z))T''
(-i(A-A^*)))\\
&=z\tilde{T}(A)
\end{align}
Thus $\tilde{T}$ has the required properties.
\end{proof}

\begin{proposition}\label{propapproxnonnormalequiv}
    A set of measurements $\hM$ (states $\hS$) is contextuality non-confirming with respect to Definition \ref{defnonnormal} if and only if $\hM$ ($\hS$) is contextuality non-confirming with respect to Definition \ref{defapprox}
\end{proposition}
\begin{proof}
    We prove the measurement case here. The state side follows from the same arguments easily by replacing $\hM \to \hO$ and $\sh \to \hS$. As a reminder, $\hO$ is the set of all observables.
    
    Assume first that $\hM$ is contextuality non-confirming with respect to Definition \ref{defnonnormal}.

Let us fix a finite set of effects $\hF\subset E(\hM)$. We now have the following especially for all $E \in \hF$ and $\rho \in \sh$.
    \begin{align}
        \tr{\rho E}=\int_\Omega T(E)(\lambda) \, d\tr{\rho G(\lambda)}
    \end{align}
    Since $T(E)\leq 1$, we can pick an everywhere bounded representative for $T(E)$. Thus by for example Lemma 4.7 in \cite{busch16} we have an increasing sequence of simple functions converging uniformly to $T(E)$ (or more precisely to the chosen representative). Let then $\varepsilon >0$ and for every $E \in \F$ choose a simple function $s_E=\sum_i a_i(E)\chi_{X_i^E}$ such that $\Vert T(E)-s_E \Vert<\varepsilon/3$. Let now $\{Y_i\}_{i=1}^n$ be some disjoint common partition (which exists, since we have finitely many $E$) for $\{X_i^E\}_{i,E \in \hF}$, i.e., for all $E$ and $i$ we have $X_i^E=\bigcup_{j \in F_i^E\subset [n]} Y_j$. Here $[n]=\{1,\dots, n\}$. Thus we may represent all the simple functions as $s_n^E=\sum_{j=1}^n a'_j(E)\chi_{Y_j}$, where $a_j'(E)=a_i(E)$ when $Y_j\subset X_i^E$. \par 
    Now $\mathrm{lin}(\hF)$ is a self-adjoint norm-closed subspace, which has the countable dense set of complex-rational linear combinations of all $E \in \hF$. Thus by Lemma \ref{lemmastateextension} there is a null set $N \subset \Omega$ and states $\{\omega_\lambda\}_{\lambda \notin N}$ such that $T(E)(\lambda)=\omega_\lambda(E)$ for all $E \in \hF$ and $\lambda \notin N$. For every $i \in [n]$, choose a $\lambda_i \in Y_i\setminus N$. These exist, since we can assume without loss of generality that every $Y_i$ has positive measure. Now the set  $\{\omega \in \lh^*\; | \; \forall E \in \hF: |\omega(E)|<\varepsilon/3 \}$ is a weak$^*$-open neighborhood of 0. Furthermore the set of normal states is dense in $\hS(\lh)$ (every Banach space is weakly$^*$ dense under canonical embedding in its double dual, see e.g. Exercise 1, Section 4 in \cite{rudin1991functional}), so for every $i \in [n]$ there exists a $\sigma_i \in \sh$ such that $\omega_{\lambda_i}-\tr{\sigma_i(\cdot)} \in V_\hF$. We now claim that the EB-channel $\Lambda_{\hF,G}^*:\lh \to \lh$ defined by 
    \begin{align}
        \Lambda_{\hF,G}^*(A)=\sum_{i=1}^n \tr{\sigma_iA}G(Y_i)
    \end{align}
    fulfills the following for all $E \in \hF$ and $\rho \in \sh$.
    \begin{align}
        |\tr{\rho E}-\tr{\rho\Lambda^*_{\hF,G}(E)}|<\varepsilon
    \end{align}
    Indeed, we have the following.
    \begin{align}
        &|\tr{\rho E}-\tr{\rho\Lambda^*_{\hF,G}(E)}| \\
        &\leq \left| \int_\Omega (T(E)(\lambda)-s_n^E(\lambda))\, d\tr{\rho G(\lambda)} \right|\\
        &+\left| \sum_{j=1}^n (a'_j(E)-T(E)(\lambda_j) )\tr{\rho G(Y_j)}\right|\\
        &+\left| \sum_{j=1}^n (\omega_{\lambda_j}(E)-\tr{\sigma_iE} )\tr{\rho G(Y_j)}\right|\\
        &\leq \Vert T(E)-s_n^E\Vert +\Vert T(E)-s_n^E\Vert \\
        &+  \sum_{j=1}^n |\omega_{\lambda_j}(E)-\tr{\sigma_iE} |\tr{\rho G(Y_j)}\\
        &<\varepsilon/3+\varepsilon/3+\varepsilon/3=\varepsilon
    \end{align}

    This proves the first direction.
    \par Let us then focus on the other direction. Let $\hF(E(\hM))$ denote the collection of finite subsets of $E(\hM)$. Then $\hI:=\hF(E(\hM)) \times \N$ is a directed set with the order relation $\leq$ defined as 
    \begin{align}
        (\hF_1,n)\leq (\hF_2,m) \Leftrightarrow \hF_1 \subset \hF_2 \; \mathrm{and} \; n\leq m,
    \end{align}
    for $\hF_1,\hF_2 \in \hF(E(\hM))$ and $n,m \in \N$. Now by assumption, there is an EB-channel $\Lambda_{G,(\hF,n)}^*:\lh \to \lh$ of the form \eqref{eqapproxeb} such that for all $E \in \hF$ and $\rho \in \sh$ we have the following.
    \begin{align}\label{eqchannelnet}
        \left|\tr{\rho(\Lambda_{G,(\hF,n)}^*(E)-E)}\right|<1/n
    \end{align}
    Thus we get a net of channels $\{\Lambda_{G,(\hF,n)}^*\}_{(\hF,n)\in \hI}$. Let us now use the explicit measure-and-prepare expansions of the $\Lambda_{G,(\hF,n)}^*$:
    \begin{align}
        &\Lambda_{G,(\hF,n)}^*(A)\\
        &=\sum_i \tr{\sigma_i^{(\hF,n)}A}G\left(X_i^{(\hF,n)}\right)\\
        &=\int_\Omega \left(\sum_i \tr{\sigma_i^{(\hF,n)}A} \chi_{X_i^{(\hF,n)}}\right) \, dG(\lambda)
    \end{align}
    For every $A \in \lh$ with $0\leq A\leq I$ we define a net of functions  $T_{(\hF,n)}(E) \in L^\infty(\Omega,\tr{\rho_0G})$ with $T_{(\hF,n)}(E):=\sum_i \tr{\sigma_i^{(\hF,n)}E} \chi_{X_i^{(\hF,n)}}$, where $\rho_0$ is a faithful state.  Then we can define the net $\{(T_{(\hF,n)}(A))_{A \in \hE(\hi)}\}_{(\hF,n) \in \hI}$ in the Cartesian product $L^{\infty}(\Omega,\tr{\rho_0G})^{\hE(\hi)}$. Now by the Banach-Alaoglu theorem \cite{rudin1991functional} the unit ball $B(L^{\infty}(\Omega,\tr{\rho_0G}))$ of $L^{\infty}(\Omega,\tr{\rho_0G})$ is compact in the weak$^*$-topology, so by Tychonoff's theorem the Cartesian product $B(L^{\infty}(\Omega,\tr{\rho_0G}))^{\hE(\hi)}$ is compact in the product-weak$^*$-topology. Since $\{(T_{(\hF,n)}(A))_{A \in \hE(\hi)}\}_{(\hF,n) \in \hI} \subset B(L^{\infty}(\Omega,\tr{\rho_0G}))^{\hE(\hi)}$, there is a convergent subnet of $\{(T_{(\hF,n)}(A))_{A \in \hE(\hi)}\}_{(\hF,n) \in \hI}$ due to compactness. Let $(T(A))_{A \in \hE(\hi)} \in B(L^{\infty}(\Omega,\tr{\rho_0G}))^{\hE(\hi)}$ be the limit of this subnet. Now since $T_{(\hF,n)}(I)=1$ for all $\hF$ and $n$, we have that $T(I)=1$. Furthermore all the functions in the net are positive, so $T(A)\geq 0$ for all effects $A$. Finally if $A,B \in \hE(\hi)$ and $A+B\leq I$, then $T(A+B)=T(A)+T(B)$ and for $0\leq t\leq 1$ $T(tA)=tT(A)$, since both of these hold for all functions in the subnet. Thus $T$ fulfills the conditions of Lemma \ref{lemmagleasonlike} and we thus find a measurement response extension in accordance to Definition \ref{defmeasresp}. Let us still denote the extension by $T$. Thus we can define a noncontextual ontological model with respect to Definition \ref{defnonnormal}: $A \mapsto \int_\Omega T(A)(\lambda) \, dG(\lambda)$.  We want to now show that with this response function we have the following for all $E \in E(\hM)$.
    \begin{align}
        E=\int_\Omega T(E)(\lambda) \, dG(\lambda)
    \end{align}
    Let now $\rho \in \sh$ be arbitrary. Then we define the linear functional $G_{\rho}:L^{\infty}(\Omega,\tr{\rho_0G}) \to \mathbb{C}$ by
    \begin{align}
        G_{\rho}(f)&:=\int_\Omega f(\lambda) \; d\tr{\rho G(\lambda)}\\
        &=\int_\Omega f(\lambda) g_\rho(\lambda)\; d\tr{\rho_0 G(\lambda)}.
    \end{align}
    Here we used the fact that since $\rho_0$ is faihtful, $\tr{\rho G}$ is absolutely continuous with respect to $\tr{\rho_0 G}$ and thus by the Radon-Nikodym theorem there is an integrable $g_\rho$ such that $\tr{\rho G(X)}=\int_X g_\rho(\lambda) \, d\tr{\rho_0 G(\lambda)}$ for all $X \in \hB(\Omega)$. Since $g_\rho \in L^1(\Omega,\tr{\rho_0G})$, we see that $G_\rho$ is a weakly$^*$-continuous linear functional. Let now $E \in E(\hM)$ and let $\varepsilon>0$. Since $G_\rho$ is weak$^*$-continuous, there is a $(\hF,n)$ as an index in the convergent subnet of $\{T_{(\hF,n)}(E)\}_{(\hF,n) \in \hI}$ such that $|G_\rho(T_{(\hF,n)}(E))-G_\rho(T(E))|<\varepsilon/2$. Since the indices of subnets are final, there is an index $(\hF_0,m)\ge (\hF,n)$ in the subnet such that $E \in \hF_0$ and $1/m<\varepsilon/2$. Thus we have the following.
    \begin{align}
        &\left|\tr{\rho E}-\int_\Omega T(E)(\lambda) \, d\tr{\rho G(\lambda)} \right|\\
        &\leq \left|\tr{\rho(\Lambda_{G,(\hF_0,m)}^*(E)-E)}\right|\\
        &+|G_\rho(T_{(\hF_0,m)}(E))-G_\rho(T(E))|\\
        &<\varepsilon/2+\varepsilon/2=\varepsilon
    \end{align}
    Since $\varepsilon>0$ was arbitrary, we have that $\tr{\rho E}=\int_\Omega T(E)(\lambda) \, d\tr{\rho G(\lambda)}$. As also $\rho \in \sh$ was arbitrary, we get the fixed point equation $E=\int_\Omega T(E)(\lambda) \, dG(\lambda)$.
\end{proof}
Note that for this proof, the passage through finite subsets of $E(\hM)$ is essential. This is due to the nature of the weak$^*$-topology $\sigma(\lh^*,\lh)$. Indeed, the idea is to approximate the non-normal states with normal states in the weak$^*$-topology and, heuristically speaking, the weak$^*$-topology can only see distances between linear functionals for finitely many effects at once.

The state side of Proposition \ref{propapproxnonnormalequiv} can also be proven more easily by using some regularities of the trace-class. We present a proof sketch of this in the Appendix.

\subsection{Sharp measurements and the non-normal definition}

The motivation to extend Definition \ref{defnormal} to the non-normal regime was due to sharp, classical measurements such as position being contextuality-confirming. Thus to bring the logic full circle, we now show that all PVMs are contextuality non-confirming with respect to the equivalent Definitions \ref{defnonnormal} and \ref{defapprox}. Specifically we use the approximate definition to show this in the following Proposition.

\begin{proposition} \label{proppvmapproximate}
Let $A:\mathcal A\to \lh$ be a PVM where $\mathcal A$ is a $\sigma$-algebra of an outcome set $\Omega$, and let $n\in \mathbb N$, $K_1,\ldots, K_n\in \mathcal A$. Then there exist an $m\in \mathbb N$, a disjoint partition $X_0,X_1,\ldots, X_m\in \mathcal  A$ of $\Omega$, and states $\sigma_0,\sigma_1,\ldots,\sigma_m$ such that $A(K_k)$ is a fixed point of the entanglement-breaking channel
$$
\Lambda(A) = \sum_{i=0}^m {\rm tr}[\sigma_i A] A(X_i)
$$
for all $k=1,\ldots,n$.
\end{proposition}
\begin{proof}
Let $X_0=\mathbb R \setminus \cup_{k=1}^n K_k$. Taking suitable intersections and differences of the sets $K_i$ we can easily construct a disjoint partition $X_1,\ldots, X_m\in \mathcal A$ of $\cup_{k=1}^n K_k$, and subsets $I_1,\ldots, I_n\subseteq \{1,\ldots,m\}$, such that for each $k$, we have $K_k= \cup_{i\in I_k} X_i$ and $X_i\cap K_k=\emptyset$ whenever $i\notin \mathcal I_k$. (For instance if $n=2$ we can take $X_0=\Omega\setminus (K_1\cup K_2)$, $X_1=K_1\setminus K_2$, $X_2=K_2\setminus K_1$, and $X_3=K_1\cap K_2$, with $I_1=\{1,3\}$, $I_2=\{2,3\}$.)

Let $J=\{i\in \{0,\ldots,m\}\mid A(X_i)\neq 0\}$. Then for each $i\in J$ pick a state $\sigma_i$ supported in $A(X_i)$, and set $\sigma_i$ to be any state if $i\notin J$.
Then we have
$$
A(K_k\cap X_i)={\rm tr}[\sigma_iA(K_k)]A(X_i)
$$
for all $k$ and $i\in J$, because for any fixed $k$, we have $K_k\cap X_i=X_i$, ${\rm tr}[\sigma_i A(K_k)]=1$ whenever $i\in I_k$, and $K_k\cap X_i=\emptyset$, ${\rm tr}[\sigma_i A(K_k)]=0$ if $i\notin I_k$ (as $A(K_k)$ is then orthogonal to $A(X_i)$).
Now for each $k$, it follows that
\begin{align*}
A(K_k) &=\sum_{i=0}^m A(K_k\cap X_i)= \sum_{i\in J} A(K_k\cap X_i)\\
&=\sum_{i\in J} {\rm tr}[\sigma_i A(K_k)] A(X_i)= \sum_{i=0}^m {\rm tr}[\sigma_i A(K_k)] A(X_i).
\end{align*}
\end{proof}

Thus according to Proposition \ref{proppvmapproximate}, all PVMs are contextuality non-confirming with respect to the approximate Definition \ref{defapprox}, and thus also with respect to Definition \ref{defnonnormal}. Especially the non-atomic position measurement of Example \ref{exposition} is no longer contextuality-confirming with respect to these definitions.

Finally we extend the result of Proposition~\ref{proppvmapproximate} to contrast 
the relationship of commuting algebras and Definition~\ref{defnormal}. Indeed, we show that if the effects of a set measurements is a subset of a commutative von Neumann algebra on a separable Hilbert space, it is contextuality non-confirming. Here, an abstract von Neumann algebra on a separable Hilbert space means that there is a von Neumann algebra isomorphism to some subalgebra of some $\lh$, with $\hi$ separable.
\begin{corollary}
    Let $N$ be a commutative von Neumann algebra on a separable Hilbert space and $\hM$ be a set of measurements such that $E(\hM) \subset  N$. Then $\hM$ is contextuality non-confirming according to Definition \ref{defnonnormal}
\end{corollary}
\begin{proof}
    $N$ is a commutative von Neumann algebra on a separable Hilbert space, so it is generated by a single self-adjoint element $A$ \cite[Proposition III.1.21]{Takesaki1979}. Let $P:\hA \to N$ be the spectral measure of $A$, so that $N$ is generated by $P$. Here $\hA$ is the $\sigma$-algebra of $P$. Since $P$ is a single PVM, by Proposition \ref{proppvmapproximate}, $P$ is contextuality non-confirming. Thus there is a $P$-measurement response $T$ such that 
    \begin{align}
        P(X)=\int_\Omega T(P(X))(\lambda) \, dP(\lambda),
    \end{align}
    for all $X \in \hA$.
    Let then $E \in E(\hM)$. Since $E \in N$, we have that $\lim_{n \to \infty} \left\Vert E-\sum_{i=1}^{k_n} a_i^n P_{i}^n \right\Vert$, where $P_i^n \in \{P(X) \; | \; X \in \hA\}$ for all $i,n$. This follows e.g. from the spectral theorem. We now especially have for all $n \in \N$ the equality
    \begin{align}
       \sum_i a_i^n P_{i}^n=\int_\Omega T\left(\sum_i a_i^n P_{i}^n\right)(\lambda) \, dP(\lambda).
    \end{align}

We then further have the following by utlizing Proposition 4.12 (c) in \cite{busch16}.
\begin{align}
    &\left\Vert E-\int_\Omega T(E)(\lambda) \, dP(\lambda) \right\Vert \\ 
    &\leq \left\Vert E- \sum_i a_i^n P_{i}^n\right\Vert\\
    &+\left\Vert \int_\Omega T\left(E-\sum_i a_i^n P_{i}^n\right)(\lambda) \, dP(\lambda) \right\Vert\\
    &\leq \left\Vert E- \sum_i a_i^n P_{i}^n\right\Vert+2\left\Vert T\left(E-\sum_i a_i^n P_{i}^n\right) \right\Vert_\infty \\
    &\leq (1+2\Vert T\Vert)\left\Vert E- \sum_i a_i^n P_{i}^n\right\Vert \to 0
\end{align}
Here we used the fact that $T$ is continuous (in fact $\Vert T\Vert=1$) as a positive linear map between $C^*$-algebras.
As $E$ was arbitrary, we have that $E(\hM)$ is contextuality non-confirming.
\end{proof}
\section{Conclusions}

We have shown that the traditional probability-theory-based definition of generalised contextuality leads to a striking effect of commuting measurements revealing the contextuality of quantum theory. The simplest example of such behaviour is the position measurement. To solve this purely infinite-dimensional discrepancy, we have proposed to extend the definition of generalised contextuality beyond probability theory. This is done by supplementing ontological models by non-sigma-additive response functions. We have shown that this extension naturally arises from a physically-motivated assumption of noncontextuality holding for all finite sets of effects. This resolves the discrepancy of the position measurement.

In the process, we have characterized the fixed points of infinite-dimensional entanglement-breaking channels in the Schrödinger picture to be commuting sets or, equivalents, sets of states having no set-coherence \cite{designolle21b}, and shown that this if furthermore equivalent to one’s inability to prove contextuality of quantum theory using such states. On the measurement side, we have shown the connection of fixed points of a Heisenberg picture EB-channel to broadcasting and further to a representation via an eigenvalue-1 parent POVM. Furthermore, with some additional assumptions the equivalence between these notions was shown, but the general case is still an open question.

For future work, the characterisation of fixed points of entanglement-breaking channels in the Heisenberg picture is left as an open question. We also believe that it will be interesting to investigate whether one can construct contextuality witnesses, i.e. inequalities on correlations, that can distinguish between the normal and the non-normal definitions. It will also be interesting to investigate the connection between noncontextual models and representability in the case of an algebra on a non-separable Hilbert space.

\begin{acknowledgments}
    \paragraph{Acknowledgments}
    We are thankful to Lauritz van Luijk and Henrik Wilming for pointing out their work \cite{Galke2024}, and for interesting discussions. We are also thankful to Yujie Zhang for interesting discussions. This work has been supported by the Swedish Research Council (grant no. 2024-05341), the Wallenberg Initiative on Networks and Quantum Information (WINQ), the Swiss National Science Foundation (Ambizione  PZ00P2-208779), the Niedersächsisches Ministerium für Wissenschaft und Kultur.
\end{acknowledgments}

\bibliographystyle{apsrev4-1}

\bibliography{references}{}

\appendix

\section{Definition \ref{defnormal} is a fixed point problem of EB-channels.}\label{appendixdefnormal}
In this Appendix we supply some mathematical details relating Definition \ref{defnormal} to fixed point problems of EB-channels. 

Now due to Pettis measurability theorem \cite[Theorem 1.1.6]{Hytönen2016}, $\lambda \mapsto \sigma_\lambda$ in Definition \ref{defnormal} is equivalently Bochner measurable i.e. the function $\lambda \mapsto \sigma_\lambda$ is measurable as a Banach-space valued function. This is in general stronger than just the "expectation values" $\lambda \mapsto \tr{\sigma_\lambda A}$ being measurable for all $A \in \lh$ (i.e. weak measurability).
Since $(\Omega,\hB(\Omega))$ is a standard Borel space, comparing Definition \ref{defnormal} and Theorem 2 of \cite{holevo2005separability}, we see that this definition induces an entanglement breaking channel, thus forming an analogue to the finite-dimensional case: 
$\tr{\rho M(X)}=\tr{\Lambda_{EB}(\rho)M(X)}=\tr{\rho \Lambda_{EB}^*(M(X))}$, with $\Lambda_{EB}(T):=\int_\Omega \sigma_\lambda \, d(\tr{T G(x)})$ for all $T \in \th$, and $\Lambda_{EB}^*$ the Heisenberg dual. Therefore, when examining full quantum theory, finding the contextuality non-confirming sets of states and measurements reduces to finding the fixed points of an EBC.

\section{Details on Example \ref{exqfunction}}
It was shown in Example \ref{exqfunction} the equality $f=f *g^n$ holds for all $n \in \N$, where $f(w)=\ip{w}{T|w}$, $g(z):=\frac{e^{-|z|^2}}{\pi}$ and $g^n=g *g*\cdots *g$ is the $n$-fold convolution. Here we show that this equality implies that $f$ is constant.

Now by direct calculation one obtains $g^{2^n}(z)=\frac{1}{2^n \pi}e^{-\frac{|z|^2}{2^n}}$. Also, since $f=f *g$, $g$ is continuous and $\Vert f \Vert\leq \Vert T \Vert$, $f$ is a bounded continuous function. Let us now regard $f$ as a function $f:\R^2 \to \R$ using the isomorphism $\mathbb{C} \simeq \R^2$ (we can assume that $T$ is a self-adjoint fixed point). We'll show that both partial derivatives of $f$ exist. Now equation \eqref{eqconvolution2} implies
   \begin{align}
       f(x,y)=\int_{\mathbb{\R^2}} \frac{e^{-(x-u)^2-(y-v)^2}}{\pi} \, f(u,v) \, du\,dv
   \end{align}
   Thus 
   \begin{align}
       &\frac{f(x+h,y)-f(x,y)}{h}\\
       &= \frac{1}{\pi}  \int_{\mathbb{\R^2}} \frac{e^{-(x+h-u)^2}-e^{-(x-u)^2}}{h}e^{-(y-v)^2}f(u,v) \,  \, du \,dv  
   \end{align}
   By the Mean Value Theorem we now have that $\frac{|e^{-(x+h-u)^2}-e^{-(x-u)^2}|}{h}=|2c(u)e^{-c(u)^2}|$ for some $c(u) \in (x-u,x-u+h)$. We can without loss of generality assume that $h<1$, so that $\frac{|e^{-(x+h-u)^2}-e^{-(x-u)^2}|}{h}\leq \sup_{z \in (x-u,x-u+1)}|2ze^{-z^2}|$. Define then the function $s:\R \to \R$ with 
   \begin{align}
       s(u):=\begin{cases}
           -2(u-x)e^{-(u-x)^2}, & u-x\le -1 \\ 
           2(u-x-1)e^{-(u-x-1)^2}, & u-x-1\ge 1 \\
           \sup_{z \in \R}|2ze^{-z^2}|, & \mathrm{otherwise}
       \end{cases}
   \end{align}
   One can easily see that $z\mapsto |2ze^{-z^2}|$ is increasing on $(\infty, -2^{-1/2}]$ and decreasing on $[2^{-1/2},\infty)$. Therefore $\frac{|e^{-(x+h-u)^2}-e^{-(x-u)^2}|}{h}\leq \sup_{z \in (x-u,x-u+1)}|2ze^{-z^2}|\leq s(u)$. Furthermore $s(u)$ is integrable, so 
   \begin{align}
      \left| \frac{e^{-(x+h-u)^2}-e^{-(x-u)^2}}{h}e^{-(y-v)^2}f(u,v)\right|\leq s(u)e^{-(y-v)^2}\Vert f \Vert 
   \end{align}
   Thus by the dominated convergence theorem $\partial_x f=f*\partial_xg$, i.e. the partial derivative exists. The $\partial_y$-case is symmetric so both partial derivatives exist. \par 
   Let us now show that $\partial_if=0$ for $i=x,y$. Now for every $n \in \N$
   \begin{align}
       |\partial_xf(x,y)|&=|\partial_x (f * g^{2^n})(x,y)| \\
       &=\frac{1}{2^n\pi}\left|\int_{\R^2} \partial_x\left(  e^{-\frac{(x-u)^2+(y-v)^2}{2^n}}\right)f(u,v) \, du \,dv \right| \\
       &\leq \frac{\Vert f \Vert}{2^{n/2}\sqrt{\pi}}\int_{\R} \frac{2}{2^n}|u-x| e^{-\frac{(u-x)^2}{2^n}} \, du  \\
       &=\frac{\Vert f \Vert}{2^{n/2}\sqrt{\pi}}\int_{\R} 2|z| e^{-z^2} \, dz \\
       &= \frac{4\Vert f \Vert}{2^{n/2}\sqrt{\pi}}\int_{0}^\infty z e^{-z^2} \, dz
   \end{align}
   Thus as $n \to \infty$, $|\partial_xf(x,y)|\to 0 $. This implies that $\partial_x f=0$, as $x,y$ were arbitrary. The case for $\partial_yf$ is symmetric, so $\partial_yf=\partial_x f=0$. Therefore $f(z)=\ip{z}{T|z}$ is a constant function. 

\section{Alternative proof for the state side of Proposition \ref{propapproxnonnormalequiv}}
Here we present a proof sketch of Proposition \ref{propapproxnonnormalequiv} for the state side using some regularities of the trace-class instead.

The direction "Def. \ref{defnonnormal} $\Rightarrow$ Def. \ref{defapprox}" can be shown the same way as in the original proof, but the states can immediately be chosen normal thanks to Proposition \ref{propdefequivalence}. 

For the direction "Def. \ref{defapprox} $\Rightarrow$ Def. \ref{defnonnormal}" one can use the fact that the set Schrödinger picture trace non-increasing measure-and-prepare channels are compact in the point-ultraweak-topology generated by the seminorms $\Lambda \mapsto |\tr{\Lambda(T)K}|$ with $T \in \th$ and $K \in \hK(\hi)$. This compactness can be seen with a standard argument using the Banach-Alaoglu theorem, Tychonoff's theorem, and the fact that the set of trace non-increasing measure-and-prepare channels are closed in this topology \cite[Supplement]{haapasalo2021operational}. Applying this to the net of channels $(\Lambda_{(\hF,n)})$ in Eq. \eqref{eqchannelnet} for all $\rho \in \hS$, we find a limit point $\tilde{\Lambda}$, for which $\tr{\tilde{\Lambda}(\rho)K}=\tr{\rho K}$ holds for all compact effects $K$ and for all $\rho \in \hS$. Thus $\tilde{\Lambda}(\rho)=\rho$ for all $\rho \in \hS$. One thus only needs to possibly modify $\tilde{\Lambda}$ to be trace-preserving. This can be done as follows. Let $\rho_0$ be a fixed point of $\tilde{\Lambda}$ with $P:=\mathrm{supp}(\rho_0)=\bigvee_{\rho \in \hS}\mathrm{supp}(\rho)$. This can be done by e.g. defining $\rho_0=\sum_n t_n \rho_n$ for some countable dense set $\{\rho_n\}$ of $\hS$ and $t_n>0$ s.t. $\sum_n t_n=1$. Then it's easy to show that $P\tilde{\Lambda}^*(P)P=P$ and we can thus define the EB-channel
\begin{align}
    \Lambda:=P\tilde{\Lambda}(P(\cdot)P)P+\tr{P^\perp(\cdot)}\sigma,
\end{align}
for some state $\sigma$. The EB-channel $\Lambda$ then has the correct fixed points.

\end{document}